\documentclass[notfinal]{acmtrans2m_calc}
\def\mypage{\thepage} 

\usepackage{amsfonts}
\usepackage{gastex}

\newtheorem{theorem}{Theorem}[section]
\newtheorem{corollary}[theorem]{Corollary}
\newtheorem{proposition}[theorem]{Proposition}
\newtheorem{lemma}[theorem]{Lemma}
\newdef{definition}[theorem]{Definition}
\newdef{remark}[theorem]{Remark}
\newdef{example}[theorem]{Example}

\newcommand{\tuple}[1]{\langle #1 \rangle}

\newcommand{\amap}{f}

\newcommand{\aset}{X}
\newcommand{\asetbis}{Y}

\newcommand{\aelem}{x}

\newcommand{\length}[1]{|#1|}

\newcommand{\aalphabet}{\Sigma}
\newcommand{\aletter}{a}
\newcommand{\aletterbis}{b}
\newcommand{\aletterter}{c}
\newcommand{\adataword}{\sigma}
\newcommand{\aword}{w}
\newcommand{\emptyword}{\varepsilon}

\newcommand{\aformula}{\phi}
\newcommand{\aformulabis}{\psi}
\newcommand{\aformulater}{\chi}

\newcommand{\nextt}{\mathtt{X}}
\newcommand{\dnext}{\bar{\mathtt{X}}}
\newcommand{\until}{\mathtt{U}}
\newcommand{\release}{\mathtt{R}}
\newcommand{\sometimes}{\mathtt{F}}
\newcommand{\always}{\mathtt{G}}

\newcommand{\nuparrow}{\not\,\uparrow}

\newcommand{\eexists}[2]{\exists #1 (#2)}
\newcommand{\fforall}[2]{\forall #1 (#2)}

\newcommand{\aregaut}{\mathcal{A}}

\newcommand{\aclass}{D}
\newcommand{\confs}{F}
\newcommand{\acaut}{\mathcal{C}}
\newcommand{\acval}{v}
\newcommand{\amachine}{\mathcal{M}}

\newcommand{\locs}{Q}
\newcommand{\aloc}{q}
\newcommand{\locsbis}{R}
\newcommand{\alocbis}{r}
\newcommand{\atransf}{\varphi}

\newcommand{\egdef}{\stackrel{\mbox{\tiny def}}{=}}

\newcommand{\equivdef}{\stackrel{\mbox{\tiny def}}{\Leftrightarrow}}

\newcommand{\ooverline}[1]{\overline{\overline{#1}}}

\title
{Safety Alternating Automata on Data Words}

\author
{RANKO LAZI\'C \\
 Department of Computer Science, University of Warwick, UK}

\begin{abstract}
A data word is a sequence of pairs of
a letter from a finite alphabet and
an element from an infinite set,
where the latter can only be compared for equality.
Safety one-way alternating automata with one register
on infinite data words are considered,
their nonemptiness is shown \textsc{ExpSpace}-complete,
and their inclusion decidable but not primitive recursive.
The same complexity bounds are obtained
for satisfiability and refinement, respectively,
for the safety fragment of linear temporal logic with freeze quantification.
Dropping the safety restriction,
adding past temporal operators,
or adding one more register,
each causes undecidability.
\end{abstract}




\begin{document}

\begin{bottomstuff}
This paper is a revised and extended version of \cite{Lazic06}.
\newline
This research was supported
by grants from the EPSRC (GR/S52759/01) and the Intel Corporation,
and by ENS Cachan.
\end{bottomstuff}

\maketitle

\section{Introduction}

\subsubsection*{Context}

Logics and automata for words and trees over finite alphabets
are relatively well-understood.  Motivated partly by
the need for formal verification and synthesis of infinite-state systems,
and the search for automated reasoning techniques for XML,
there is an active and broad research programme on
logics and automata for words and trees which have richer structure.

Segoufin's survey \cite{Segoufin06} is a summary of
the substantial progress made on reasoning about data words and data trees.
A data word is a word over a finite alphabet,
with an equivalence relation on word positions.
Implicitly, every word position is labelled by an element (``datum'')
from an infinite set (``data domain''),
but since the infinite set is equipped only with the equality predicate,
it suffices to know which word positions are labelled by equal data,
and that is what the equivalence relation represents.
Similarly, a data tree is a tree (countable, unranked and ordered)
whose every node is labelled by a letter from a finite alphabet,
with an equivalence relation on the set of its nodes.

It has been nontrivial to find satisfactory
specification formalisms even for data words.
First-order logic was considered in \cite{Bojanczyketal06a,David04},
and related automata were studied further in \cite{Bjorklund&Schwentick07}.
The logic has variables which range over word positions
($\{0, \ldots, l - 1\}$ or $\mathbb{N}$),
a unary predicate for each letter from the finite alphabet,
and a binary predicate $x \sim y$ for
the equivalence relation that represents equality of data labels.
FO$^2(\sim, <, +1)$ denotes such a logic with two variables
and binary predicates $x + 1 = y$ and $x < y$.
Over finite and over infinite data words,
satisfiability for FO$^2(\sim, <, +1)$ was proved decidable
and at least as hard as reachability for Petri nets \cite{Bojanczyketal06a}.
The latter problem is \textsc{ExpSpace}-hard \cite{Lipton76},
but its elementarity is still an open question.
Elementary complexity of satisfiability can be obtained at the price of
substantially reducing the navigational power: over finite data words,
\textsc{NExpTime}-completeness for FO$^2(\sim, <)$
was established in \cite{David04} and
$3$\textsc{NExpTime}-membership for FO$^2(\sim, +1)$
follows from \cite{Bojanczyketal06b}.
In the other direction, if FO$^2(\sim, <, +1)$ is extended by one more variable,
$+1$ becomes expressible using $<$, but satisfiability was shown undecidable
already for FO$^3(\sim, +1)$ \cite{Bojanczyketal06a}.

An alternative approach to reasoning about data words
is based on automata with registers \cite{Kaminski&Francez94}.
A register is used for storing a datum for later equality comparisons
(i.e.\ an equivalence class for later membership testing).
Nonemptiness of one-way nondeterministic register automata
over finite data words has relatively low complexity:
\textsc{NP}-complete \cite{Sakamoto&Ikeda00} or
\textsc{PSpace}-complete \cite{Demri&Lazic09},
depending on technical details of their definition.
Unfortunately, such automata fail to provide a satisfactory notion
of regular language of finite data words,
as they are not closed under complement \cite{Kaminski&Francez94}
and their nonuniversality is undecidable \cite{Neven&Schwentick&Vianu04}.
To overcome those limitations,
one-way alternating automata with $1$ register (for short, 1ARA$_1$)
were proposed in \cite{Demri&Lazic09}:
they are closed under Boolean operations,
their nonemptiness over finite data words is decidable,
and future-time fragments of temporal logics such as
LTL or the modal $\mu$-calculus extended by $1$ register
are easily translatable to such automata.
However, nonemptiness for 1ARA$_1$ turned out to be
not primitive recursive over finite data words,
and undecidable (more precisely, $\Pi^0_1$-hard) over infinite ones
with the weak acceptance mechanism \cite{Muller&Saoudi&Schupp86}
and thus also with B\"uchi or co-B\"uchi acceptance.

\subsubsection*{Contribution}

We consider one-way alternating automata with $1$ register
with the safety acceptance mechanism over infinite data words
(i.e.\ data $\omega$-words).
The languages of such automata are safety properties \cite{Alpern&Schneider87}:
every rejected data $\omega$-word has a finite prefix such that
every other data $\omega$-word which extends it is also rejected.
(Over finite data words, safety is not a restriction.)

The main result is that
nonemptiness of safety 1ARA$_1$ is in \textsc{ExpSpace}.
We say that a sentence of LTL is safety iff each occurrence of
the `until' operator is under an odd number of negations.
In particular, each `eventually' (resp., `always') must be
under an odd (resp., even) number of negations.
By showing that the safety fragment of future-time LTL with $1$ register
is translatable in logarithmic space to safety 1ARA$_1$,
and that satisfiability for the fragment is \textsc{ExpSpace}-hard,
we conclude \textsc{ExpSpace}-completeness of both problems.

The \textsc{ExpSpace} upper bound is surprising
since even decidability is fragile:
by \cite[Theorem~5.2]{Demri&Lazic09},
satisfiability for future-time LTL with $1$ register
on data $\omega$-words is $\Pi^0_1$-hard,
and from the proof of \cite[Theorem~5.4]{Demri&Lazic09},
the same is true for the safety fragment
if past temporal operators or one more register are added
(cf.\ related undecidability results in
\cite{Neven&Schwentick&Vianu04,David04}).
Moreover, nonemptiness of safety forward (i.e.\ downward and rightward)
alternating automata with $1$ register on data trees was shown
decidable but not elementary \cite{Jurdzinski&Lazic07}.
Another setting where decidability \cite{Ouaknine&Worrell06b}
was obtained by restricting to safety sentences is
that of metric temporal logic on timed $\omega$-words,
but the complexity is again not elementary \cite{Bouyeretal08}.

The proof of \textsc{ExpSpace}-membership is in two stages.
The first consists of translating a given safety 1ARA$_1$ $\aregaut$
to a nondeterministic automaton with faulty counters $\acaut_\aregaut$
which is on $\omega$-words over the alphabet of $\aregaut$ and
which is nonempty iff $\aregaut$ is.
The counters of $\acaut_\aregaut$ are faulty in the sense that
they are subject to incrementing errors,
i.e.\ they can spontaneously increase at any time.
Although a nonemptiness-preserving translation
from 1ARA$_1$ with weak acceptance to counter automata with incrementing errors
was given in \cite{Demri&Lazic09}, applying it to safety 1ARA$_1$
produces automata with the B\"uchi acceptance mechanism,
where the latter ensures that certain loops
cannot repeat infinitely due to incrementing errors.
To obtain safety automata, we enrich the instruction set by
\emph{nondeterministic transfers}.
When applied to a counter $c$ and a set of counters $C$,
such an instruction transfers the value of $c$ to the counters in $C$,
nondeterministically splitting it.
Thus we obtain $\acaut_\aregaut$ whose nonemptiness amounts to
existence of an infinite computation from the initial state.
However, a further observation on the resulting automata is required:
the counters of such an automaton are nonempty subsets of a certain set
(essentially, the set of states of the given safety 1ARA$_1$),
and it suffices to use nondeterministic transfers
which are simultaneous for all counters and
which have a certain distributivity property in terms of
the partial-order structure of the set of all counters.

The second stage of the proof is then an inductive counting argument
which shows that $\acaut_\aregaut$ is nonempty iff
it has a computation from the initial state
of length doubly exponential in the size of $\aregaut$.
Some of the techniques are also used in the proof that
termination of channel machines with occurrence testing and insertion errors
is primitive recursive \cite{Bouyeretal08}.
Although counters are simpler resources than channels,
the class of machines considered do not have instructions
which correspond to the nondeterministic transfers,
and the sets of channels and messages
(which are counterparts to the sets of counters)
have no special structure.

We also show that language inclusion between two safety 1ARA$_1$ is decidable,
and hence that refinement (i.e., validity of implication) between
two sentences of safety future-time LTL with $1$ register is also decidable.
Since the safety fragment is closed under conjunctions and disjunctions,
it follows that satisfiability is decidable for
Boolean combinations of safety sentences.
The latter is thus a competing logic to
FO$^2(\sim, <, +1)$ on data $\omega$-words.
They are incomparable in expressiveness:
there exist properties involving the past
(e.g.\ `every $\aletterbis$ is preceded by an $\aletter$ with the same datum')
which are expressible in FO$^2(\sim, <, +1)$
but not by a Boolean combination of safety sentences
(not even in future-time LTL with $1$ register),
and the reverse is true of some constraints
involving more than $2$ word positions
(e.g.\ `whenever $\aletter$ is followed by $\aletterbis$ with the same datum,
$\aletterter$ does not occur in between').
However, as pointed out above, it is not known whether
satisfiability for FO$^2(\sim, <, +1)$ is elementary,
whereas we establish that already satisfiability for
negations of safety sentences is not primitive recursive,
and hence also universality for safety 1ARA$_1$.

\section{Preliminaries}

In this section, we define
safety one-way alternating automata and
safety future-time linear temporal logic
with $1$ register on data $\omega$-words,
as well as the class of counter automata that will be used in
the proof of \textsc{ExpSpace}-membership in Section~\ref{s:upper}.
We also show some of their basic properties,
in particular a logarithmic-space translation
from the linear temporal logic to the alternating automata.

\subsection{Data Words}

A \emph{data $\omega$-word} $\adataword$ over a finite alphabet $\aalphabet$ is
an $\omega$-word $\mathrm{str}(\adataword)$ over $\aalphabet$ together with
an equivalence relation $\sim^\adataword$ on $\mathbb{N} = \{0, 1, \ldots\}$.
We write $\mathbb{N} / \sim^\adataword$ for
the set of all classes of $\sim^\adataword$.
For $i \in \mathbb{N}$, we write
$\adataword(i)$ for the letter at position $i$, and
$[i]_{\sim^\adataword}$ for the class that contains $i$.
When $\adataword$ is understood, we may write
simply $\sim$ instead of $\sim^\adataword$.
We shall sometimes refer to classes of $\sim$ as `data'.

In some places, we shall also need the concept of a finite data word.
For $i > 0$, the $i$-prefix of a data $\omega$-word $\adataword$ is
the finite data word whose letters are $\adataword(0) \cdots \adataword(i - 1)$
and whose equivalence relation is
$\sim^\adataword$ restricted to $\{0, \ldots, i - 1\}$.

\subsection{Register Automata}

The definition of safety one-way alternating $1$-register automata below
is based on the more general one of
weak two-way alternating register automata in \cite{Demri&Lazic09}.
A configuration of such an automaton
at a position $i$ of a data $\omega$-word $\adataword$
will consist of one of finitely many automaton states
and a register value $\aclass \,\in\, \mathbb{N} / \sim$.
From it, depending on the state, the letter $\adataword(i)$,
and whether $\aclass = [i]_\sim$ (denoted $\uparrow$)
or $\aclass \neq [i]_\sim$ (denoted $\nuparrow$),
the automaton chooses a pair $\locs', \locs'_\downarrow$ of sets of states.
The resulting set of configurations at the next word position is
$\{\tuple{\aloc', \aclass} \,:\, \aloc' \in \locs'\} \,\cup\,
 \{\tuple{\aloc', [i]_\sim} \,:\, \aloc' \in \locs'_\downarrow\}$,
i.e.\ the states in $\locs'$ are associated with the old register value,
and the states in $\locs'_\downarrow$ with the class of position $i$.
Following \cite{Brzozowski&Leiss80},
what choices of pairs of sets of states are possible
will be specified in each case by a positive Boolean formula.
That formalisation, in contrast to listing all possible such choices,
will enable a logarithmic-space translation from
safety future-time LTL with $1$ register.

An infinite run of the automaton will consist, for each $j \in \mathbb{N}$,
of a set $\confs_j$ of all configurations at position $j$.
For each $j$, $\confs_{j + 1}$ will be
the union of some sets of configurations
chosen as above for each configuration in $\confs_j$.
Hence, a configuration will be rejecting when
its set of possible choices is empty,
and it will be accepting when it can choose
$\locs' = \locs'_\downarrow = \emptyset$.
The definition of infinite runs will ensure that
they cannot contain rejecting configurations,
so the safety acceptance mechanism will amount to
each infinite run being considered accepting.

Formally, for a finite set $\locs$,
let $\downarrow \locs = \{\downarrow \aloc \,:\, \aloc \in \locs\}$,
and let $\mathcal{B}^+_\downarrow(\locs)$ denote the set of all
positive Boolean formulae over $\locs \,\cup\, \downarrow \locs$,
where we assume that $\locs$ and $\downarrow \locs$ are disjoint:
\[\atransf \:::=\:
  \aloc \,\mid\, {\downarrow} \aloc \,\mid\,
  \top \,\mid\, \bot \,\mid\,
  \atransf \wedge \atransf \,\mid\, \atransf \vee \atransf\]

A \emph{safety one-way alternating automaton with $1$ register}
(shortly, \emph{safety 1ARA$_1$}) $\aregaut$ is a tuple
$\tuple{\aalphabet, \locs, \aloc_I, \delta}$ such that:
\begin{itemize}
\item
$\aalphabet$ is a finite alphabet;
\item
$\locs$ is a finite set of states, and
$\aloc_I \in \locs$ is the initial state;
\item
$\delta:
 (\locs \times \aalphabet \times \{\uparrow, \nuparrow\}) \rightarrow
 \mathcal{B}^+_\downarrow(\locs)$
is a transition function.
\end{itemize}

Satisfaction of a positive Boolean formula over
$\locs \,\cup\, \downarrow \locs$ by a pair of sets
$\locs', \locs'_\downarrow \subseteq \locs$
is defined by structural recursion:
\[\begin{array}{rcl@{\hspace{2em}}rcl}
\locs', \locs'_\downarrow \,\models\, \aloc
& \equivdef &
\aloc \in \locs'
&
\locs', \locs'_\downarrow \,\models\, \top
& & \\
\locs', \locs'_\downarrow \,\models\, {\downarrow} \aloc
& \equivdef &
\aloc \in \locs'_\downarrow
&
\locs', \locs'_\downarrow \,\not\models\, \bot
& &
\end{array}\]
\[\begin{array}{rcl}
\locs', \locs'_\downarrow \,\models\, \atransf \wedge \atransf'
& \equivdef &
\locs', \locs'_\downarrow \,\models\, \atransf \ \mathrm{and}\
\locs', \locs'_\downarrow \,\models\, \atransf'
\\
\locs', \locs'_\downarrow \,\models\, \atransf \vee \atransf'
& \equivdef &
\locs', \locs'_\downarrow \,\models\, \atransf \ \mathrm{or}\
\locs', \locs'_\downarrow \,\models\, \atransf'
\end{array}\]

A configuration of $\aregaut$ for
a data word $\adataword$ is an element of
$\locs \times (\{j \,:\, 0 \leq j < \length{\adataword}\} / \sim)$.
For a position $0 \leq i < \length{\adataword}$,
and finite sets $\confs$ and $\confs'$ of configurations,
we write $\confs \stackrel{\adataword, i}{\longrightarrow} \confs'$
iff, for each $\tuple{\aloc, \aclass} \in \confs$,
there exist
$\locs^{\tuple{\aloc, \aclass}},
 \locs^{\tuple{\aloc, \aclass}}_\downarrow
 \subseteq \locs$
which satisfy the formula $\delta(\aloc, \adataword(i), \uparrow)$
if $\aclass = [i]_\sim$,
or the formula $\delta(\aloc, \adataword(i), \nuparrow)$
if $\aclass \neq [i]_\sim$,
such that
\[\confs' \:=\:
  \{\tuple{\aloc', \aclass} \::\:
    \tuple{\aloc, \aclass} \in \confs \,\wedge\,
    \aloc' \in \locs^{\tuple{\aloc, \aclass}}\} \,\cup\,
  \{\tuple{\aloc', [i]_\sim} \::\:
    \tuple{\aloc, \aclass} \in \confs \,\wedge\,
    \aloc' \in \locs^{\tuple{\aloc, \aclass}}_\downarrow\}\]

We say that $\aregaut$ accepts
a data $\omega$-word $\adataword$ over $\aalphabet$
iff it has an infinite run
$\confs_0 \stackrel{\adataword, 0}{\longrightarrow}
 \confs_1 \stackrel{\adataword, 1}{\longrightarrow}
 \cdots$
where $\confs_0 = \{\tuple{\aloc_I, [0]_\sim}\}$
consists of the initial configuration.
We write $\mathrm{L}(\aregaut)$ for the language of $\aregaut$,
i.e.\ the set of all data $\omega$-words over $\aalphabet$
that $\aregaut$ accepts.

\begin{example}
\label{ex:RA}
A safety 1ARA$_1$ with alphabet $\{\aletter, \aletterbis, \aletterter\}$
and three states is depicted in Figure~\ref{f:RA}.
It rejects a data $\omega$-word iff
there is an occurrence of $\aletter$,
a subsequent occurrence of $\aletterbis$ with the same datum,
and an occurrence of $\aletterter$ between them.

The automaton is deterministic, except for the universal branching
from state $\aloc$ at letter $\aletter$.
When behaviour does not depend on whether
the class in the register equals the class of the current position,
the two cases are not shown separately.
In particular, we have
$\delta(\aloc, \aletter, \uparrow) =
 \delta(\aloc, \aletter, \nuparrow) =
 \aloc \,\wedge\, {\downarrow} \aloc'$.
The absence of a transition from $\aloc''$
labelled by $\aletterbis$ and $\uparrow$
means that we have rejection in that case, i.e.\
$\delta(\aloc'', \aletterbis, \uparrow) = \bot$.
\end{example}

\begin{narrowfig}{.725\textwidth}
\setlength{\unitlength}{4em}
\begin{picture}(6,2)(.25,1)
\gasset{Nadjust=wh}
\node[Nmarks=i,iangle=180](1)(1,2){$\aloc$}
\node[Nadjustdist=0](11)(2,2){}
\node(2)(3,2){$\aloc'$}
\node(3)(5,2){$\aloc''$}
\drawloop[loopdiam=.75,loopangle=90](1){$\aletterbis$}
\drawloop[loopdiam=.75,loopangle=270](1){$\aletterter$}
\drawedge[AHnb=0](1,11){$\aletter$}
\drawedge[curvedepth=.25](11,1){}
\drawedge(11,2){$\downarrow$}
\drawloop[loopdiam=.75,loopangle=90](2){$\aletter$}
\drawloop[loopdiam=.75,loopangle=270](2){$\aletterbis$}
\drawedge(2,3){$\aletterter$}
\drawloop[loopdiam=.75,loopangle=90](3){$\aletter$}
\drawloop[loopdiam=.75,loopangle=0](3)
{$\begin{array}{c}
  \aletterbis \\
  \nuparrow
  \end{array}$}
\drawloop[loopdiam=.75,loopangle=270](3){$\aletterter$}
\end{picture}
\caption{A register automaton}
\label{f:RA}
\end{narrowfig}

A set $L$ of data $\omega$-words over an alphabet $\aalphabet$
is called \emph{safety} \cite{Alpern&Schneider87} iff
it is closed under limits of finite prefixes,
i.e.\ for each data $\omega$-word $\adataword$,
if for each $i > 0$ there exists $\adataword'_i \in L$
with the $i$-prefixes of $\adataword$ and $\adataword'_i$ equal,
then $\adataword \in L$.%
\footnote{Hence, a set is safety iff it is closed with respect to
the Cantor metric, where the distance between two words is inversely
proportional to the length of their longest common prefix.}

\begin{proposition}
\label{pr:safety.RA}
The language of each safety 1ARA$_1$ is safety.
\end{proposition}

\begin{proof}
Suppose that $\aregaut$ is a safety 1ARA$_1$,
and for each $i > 0$ there exists $\adataword'_i \in \mathrm{L}(\aregaut)$
such that the $i$-prefixes of $\adataword$ and $\adataword'_i$ are equal.
For each $i$, let
$\confs'_{i, 0} \stackrel{\adataword'_i, 0}{\longrightarrow}
 \confs'_{i, 1} \stackrel{\adataword'_i, 1}{\longrightarrow}
 \ldots$
be an infinite run of $\aregaut$ with
$\confs'_{i, 0} = \{\tuple{\aloc_I, [0]_{\sim^{\adataword'_i}}}\}$.
For each $0 \leq j \leq i$,
let $\confs^\dag_{i, j}$ be obtained from $\confs'_{i, j}$
by replacing each class $\aclass'$ of $\adataword'_i$
with the class $\aclass$ of $\adataword$ such that
$\aclass' \cap \{0, \ldots, i - 1\} = \aclass \cap \{0, \ldots, i - 1\}$.
Now, consider the tree formed by all the sequences
$\tuple{\confs^\dag_{i, j} \,:\, 0 \leq j \leq i}$ for $i > 0$.
The tree is finitely branching, so by K\"onig's Lemma,
it contains an infinite path $\tuple{\confs_j \,:\, j \in \mathbb{N}}$.
It remains to observe that
$\confs_0 \stackrel{\adataword, 0}{\longrightarrow}
 \confs_1 \stackrel{\adataword, 1}{\longrightarrow}
 \ldots$
and $\confs_0 = \{\tuple{\aloc_I, [0]_{\sim^\adataword}}\}$.
\end{proof}

Given safety 1ARA$_1$ $\aregaut_1$ and $\aregaut_2$ with alphabet $\aalphabet$,
it is easy to construct an automaton which recognises
$\mathrm{L}(\aregaut_1) \cap \mathrm{L}(\aregaut_2)$
(resp., $\mathrm{L}(\aregaut_1) \cup \mathrm{L}(\aregaut_2)$).
It suffices to form a disjoint union of $\aregaut_1$ and $\aregaut_2$,
and add a new initial state $\aloc_I$ such that
$\delta(\aloc_I, \aletter, ?) =
 \delta(\aloc_I^1, \aletter, ?) \wedge
 \delta(\aloc_I^2, \aletter, ?)$
(resp.,
$\delta(\aloc_I, \aletter, ?) =
 \delta(\aloc_I^1, \aletter, ?) \vee
 \delta(\aloc_I^2, \aletter, ?)$)
for each $\aletter \in \aalphabet$ and $? \in \{\uparrow, \nuparrow\}$,
where $\aloc_I^1$ and $\aloc_I^2$ are the initial states of
$\aregaut_1$ and $\aregaut_2$.
We thus obtain:

\begin{proposition}
Safety 1ARA$_1$ are closed under finite intersections and finite unions,
in logarithmic space.
\end{proposition}

\subsection{Linear Temporal Logic}

Safety LTL$^\downarrow_1(\nextt, \release)$ will denote
the safety fragment of future-time linear temporal logic with $1$ register,
whose syntax is given below.
Each formula is over a finite alphabet $\aalphabet$,
over which the atomic formulae $\aletter$ range.
By restricting ourselves to formulae in negation normal form,
the safety restriction amounts to the `release' temporal operator
being available instead of its dual `until'.
The formulae may also contain the `next' temporal operator.
A freeze quantification ${\downarrow} \aformula$ binds
each free occurrence of $\uparrow$ in $\aformula$.
Such an occurrence will evaluate to true iff
the word position at the time of the freeze quantification and
the word position when the occurrence of $\uparrow$ is evaluated
are in the same class.
\[
\aformula \:::=\:
\aletter \,\mid\, \top \,\mid\, \bot \,\mid\,
\aformula \wedge \aformula \,\mid\, \aformula \vee \aformula \,\mid\,
\nextt \aformula \,\mid\, \aformula \release \aformula \,\mid\,
{\downarrow} \aformula \,\mid\, {\uparrow} \,\mid\, {\nuparrow}
\]
The `always' temporal operator can be introduced by regarding
$\always \aformula$ as an abbreviation for $\bot \release \aformula$.

For a data $\omega$-word $\adataword$ over a finite alphabet $\aalphabet$,
a position $i \in \mathbb{N}$,
a register value $\aclass \,\in\, \mathbb{N} / \sim$, and
a formula $\aformula$ over $\aalphabet$,
writing $\adataword, i \,\models_\aclass\, \aformula$ will mean that
$\aformula$ is satisfied by $\adataword$ at $i$ with respect to $\aclass$.
The satisfaction relation is defined as follows,
where we omit the Boolean cases.
\begin{eqnarray*}
\adataword, i \,\models_\aclass\, \aletter
& \equivdef &
\adataword(i) = \aletter
\\
\adataword, i \,\models_\aclass\, \nextt \aformula
& \equivdef &
\adataword, i + 1 \,\models_\aclass\, \aformula
\\
\adataword, i \,\models_\aclass\, \aformula \release \aformulabis
& \equivdef &
\mathrm{either}\ \mathrm{for\ all}\ k \geq i,\
\adataword, k \,\models_\aclass\, \aformulabis,\
\mathrm{or}\ \mathrm{for\ some}\ j \geq i, \\ & &
\adataword, j \,\models_\aclass\, \aformula\ \mathrm{and}\
\mathrm{for\ all}\ k \in \{i, \ldots, j\},\
\adataword, k \,\models_\aclass\, \aformulabis
\\
\adataword, i \,\models_\aclass\, {\downarrow} \aformula
& \equivdef &
\adataword, i \,\models_{[i]_\sim}\, \aformula
\\
\adataword, i \,\models_\aclass\, {\uparrow}
& \equivdef &
i \in \aclass
\\
\adataword, i \,\models_\aclass\, {\nuparrow}
& \equivdef &
i \notin \aclass
\end{eqnarray*}
If $\aformula$ is a sentence, i.e.\ contains no free occurrence of $\uparrow$,
we may omit $\aclass$ since it is irrelevant and write
$\adataword, i \,\models\, \aformula$.
Let $\mathrm{L}(\aformula)$ denote the language of $\aformula$,
i.e.\ the set of all data $\omega$-words over $\aalphabet$
such that $\adataword, 0 \,\models\, \aformula$.

\begin{example}
\label{ex:LTL}
Consider the following sentence $\aformula$
over alphabet $\{\aletter, \aletterbis, \aletterter\}$:
\[\always (\aletterbis \vee \aletterter \vee
    {\downarrow} \nextt \always (\aletter \vee \aletterbis \vee
      \nextt \always (\aletter \vee \aletterter \vee {\nuparrow})))\]
We have $\adataword, 0 \,\models\, \aformula$ iff,
for each occurrence of $\aletter$ in $\adataword$
and each later occurrence of $\aletterter$,
there is no later still occurrence of $\aletterbis$
with the same datum as the occurrence of $\aletter$,
i.e.\ iff $\adataword$ is accepted by the automaton in Example~\ref{ex:RA}.
\end{example}

\begin{theorem}
\label{th:LTL2RA}
For each sentence $\aformula$ of
safety LTL$^\downarrow_1(\nextt, \release)$,
a safety 1ARA$_1$ $\aregaut_\aformula$
with the same alphabet and
$\mathrm{L}(\aformula) =
 \mathrm{L}(\aregaut_\aformula)$
is computable in logarithmic space.
\end{theorem}

\begin{proof}
The translation is a straightforward adaptation of
the classical one from LTL to alternating automata
(cf.\ e.g.\ \cite{Vardi96}).

To define $\aregaut_\aformula$ with alphabet $\aalphabet$ of $\aformula$,
let the set of states $\locs$ consist of all $\aloc_{\aformula'}$
such that $\aformula'$ is either $\aformula$,
or $\aformulabis$ for a subformula $\nextt \aformulabis$ of $\aformula$,
or a subformula $\aformulabis \release \aformulater$ of $\aformula$.
Let the initial state be $\aloc_\aformula$.
The transition function is obtained by restricting to $\locs$
the function defined below by structural recursion over the set of all
$\aloc_{\aformula'}$ where $\aformula'$ is a subformula of $\aformula$.
The dual cases are omitted, and $?$ ranges over $\{\uparrow, \nuparrow\}$.
In the formula for ${\downarrow} \aformulabis$,
each occurrence of a state $\aloc'$ without $\downarrow$
is substituted by ${\downarrow} \aloc'$.
\[\begin{array}{rcl@{\hspace{2em}}rcl}
\delta(\aloc_\aletter, \aletter, ?)
& \egdef &
\top
&
\delta(\aloc_{\aformulabis \wedge \aformulater}, \aletter, ?)
& \egdef &
\delta(\aloc_\aformulabis, \aletter, ?) \wedge
\delta(\aloc_\aformulater, \aletter, ?)
\\
\delta(\aloc_\aletter, \aletter', ?)
& \egdef &
\bot, \mathrm{for}\ \aletter' \neq \aletter
&
\delta(\aloc_{\nextt \aformulabis}, \aletter, ?)
& \egdef &
\aloc_\aformulabis
\\
\delta(\aloc_\top, \aletter, ?)
& \egdef &
\top
&
\delta(\aloc_{\aformulabis \release \aformulater}, \aletter, ?)
& \egdef &
\delta(\aloc_\aformulater, \aletter, ?) \wedge
(\delta(\aloc_\aformulabis, \aletter, ?) \vee
 \aloc_{\aformulabis \release \aformulater})
\\
\delta(\aloc_{\uparrow}, \aletter, \uparrow)
& \egdef &
\top
&
\delta(\aloc_{{\downarrow} \aformulabis}, \aletter, ?)
& \egdef &
\delta(\aloc_\aformulabis, \aletter, \uparrow)
[{\downarrow} \aloc' / \aloc' \,:\, \aloc' \in \locs]
\\
\delta(\aloc_{\uparrow}, \aletter, \nuparrow)
& \egdef &
\bot
& & &
\end{array}\]

That $\aregaut_\aformula$ is computable in logarithmic space follows by
observing that, for each subformula $\aformula'$ of $\aformula$,
$\aletter \in \aalphabet$, and $? \in \{\uparrow, \nuparrow\}$,
a single traversal of $\aformula'$ suffices for computing
$\delta(\aloc_{\aformula'}, \aletter, ?)$.

Equality of the languages of $\aformula$ and $\aregaut_\aformula$
is implied by the following claim:
for each subformula $\aformula'$ of $\aformula$,
data $\omega$-word $\adataword$ over $\aalphabet$,
position $i \in \mathbb{N}$, and
register value $\aclass \,\in\, \mathbb{N} / \sim$,
we have $\adataword, i \,\models_\aclass\, \aformula'$
iff, for some $\locs', \locs'_\downarrow \subseteq \locs$ such that
$\locs', \locs'_\downarrow \,\models\,
 \delta(\aloc_{\aformula'}, \adataword(i), \uparrow)$
if $\aclass = [i]_\sim$, or such that
$\locs', \locs'_\downarrow \,\models\,
 \delta(\aloc_{\aformula'}, \adataword(i), \nuparrow)$
if $\aclass \neq [i]_\sim$,
$\aregaut_\aformula$ has an infinite run
from position $i + 1$ of $\adataword$, starting with
\[\{\tuple{\aloc', \aclass} \,:\, \aloc' \in \locs'\} \,\cup\,
  \{\tuple{\aloc', [i]_\sim} \,:\, \aloc' \in \locs'_\downarrow\}\]
(If $\aloc_{\aformula'}$ is a state of $\aregaut_\aformula$,
the latter is equivalent to $\aregaut_\aformula$ having a run
from position $i$ of $\adataword$, starting with
$\{\tuple{\aloc_{\aformula'}, \aclass}\}$.)
The claim is provable by structural induction on $\aformula'$.
We treat explicitly the two interesting cases:
$\aformula' = \aformulabis \release \aformulater$ and
$\aformula' = {\downarrow} \aformulabis$.

Suppose $\adataword, i \,\models_\aclass\, \aformulabis \release \aformulater$.
If $\adataword, j \,\models_\aclass\, \aformulabis$
for some $j \geq i$,
and $\adataword, k \,\models_\aclass\, \aformulater$
for all $k \in \{i, \ldots, j\}$,
then by the inductive hypothesis:
\begin{itemize}
\item[(i)]
for some $\locs', \locs'_\downarrow \subseteq \locs$ such that
$\locs', \locs'_\downarrow \,\models\,
 \delta(\aloc_\aformulabis, \adataword(j), \uparrow)$
if $\aclass = [j]_\sim$, or such that
$\locs', \locs'_\downarrow \,\models\,
 \delta(\aloc_\aformulabis, \adataword(j), \nuparrow)$
if $\aclass \neq [j]_\sim$,
$\aregaut_\aformula$ has an infinite run
$\confs'_{j + 1} \stackrel{\adataword, j + 1}{\longrightarrow}
 \confs'_{j + 2} \stackrel{\adataword, j + 2}{\longrightarrow}
 \cdots$
with
\[\confs'_{j + 1} =
  \{\tuple{\aloc', \aclass} \,:\, \aloc' \in \locs'\} \,\cup\,
  \{\tuple{\aloc', [j]_\sim} \,:\, \aloc' \in \locs'_\downarrow\}\]
\item[(ii)]
for all $k \in \{i, \ldots, j\}$,
for some $\locs^k, \locs^k_\downarrow \subseteq \locs$ such that
$\locs^k, \locs^k_\downarrow \,\models\,
 \delta(\aloc_\aformulater, \adataword(k), \uparrow)$
if $\aclass = [k]_\sim$, or such that
$\locs^k, \locs^k_\downarrow \,\models\,
 \delta(\aloc_\aformulater, \adataword(k), \nuparrow)$
if $\aclass \neq [k]_\sim$,
$\aregaut_\aformula$ has an infinite run
$\confs^k_{k + 1} \stackrel{\adataword, k + 1}{\longrightarrow}
 \confs^k_{k + 2} \stackrel{\adataword, k + 2}{\longrightarrow}
 \cdots$
with
\[\confs^k_{k + 1} =
  \{\tuple{\aloc', \aclass} \,:\, \aloc' \in \locs^k\} \,\cup\,
  \{\tuple{\aloc', [k]_\sim} \,:\, \aloc' \in \locs^k_\downarrow\}\]
\end{itemize}
Letting
$\confs^\dag_l =
 \{\tuple{\aloc_{\aformulabis \release \aformulater}, \aclass}\} \,\cup\,
 \bigcup_{k \in \{i, \ldots, l - 1\}} \confs^k_l$
for each $l \in \{i, \ldots, j\}$, and
$\confs^\dag_l =
 \bigcup_{k \in \{i, \ldots, j\}} \confs^k_l \,\cup\,
 \confs'_l$
for each $l \geq j + 1$,
we have by (i) and (ii) that
$\confs^\dag_i \stackrel{\adataword, i}{\longrightarrow}
 \confs^\dag_{i + 1} \stackrel{\adataword, i + 1}{\longrightarrow}
 \cdots$
and
$\confs^\dag_i =
 \{\tuple{\aloc_{\aformulabis \release \aformulater}, \aclass}\}$,
as required.
If $\adataword, k \,\models_\aclass\, \aformulater$ for all $k \geq i$,
the argument is simpler.

For the converse, suppose $\aregaut_\aformula$ has an infinite run
$\confs^\dag_i \stackrel{\adataword, i}{\longrightarrow}
 \confs^\dag_{i + 1} \stackrel{\adataword, i + 1}{\longrightarrow}
 \cdots$
with
$\confs^\dag_i =
 \{\tuple{\aloc_{\aformulabis \release \aformulater}, \aclass}\}$.
If there exists $j \geq i$ with
$\tuple{\aloc_{\aformulabis \release \aformulater}, \aclass} \notin
 \confs^\dag_{j + 1}$,
consider the minimum such $j$.
Since
$\delta(\aloc_{\aformulabis \release \aformulater}, \aletter, ?) =
 \delta(\aloc_\aformulater, \aletter, ?) \wedge
 (\delta(\aloc_\aformulabis, \aletter, ?) \vee
  \aloc_{\aformulabis \release \aformulater})$,
we obtain:
\begin{itemize}
\item[(iii)]
for some $\locs', \locs'_\downarrow \subseteq \locs$ such that
$\locs', \locs'_\downarrow \,\models\,
 \delta(\aloc_\aformulabis, \adataword(j), \uparrow)$
if $\aclass = [j]_\sim$, or such that
$\locs', \locs'_\downarrow \,\models\,
 \delta(\aloc_\aformulabis, \adataword(j), \nuparrow)$
if $\aclass \neq [j]_\sim$,
we have 
\[\{\tuple{\aloc', \aclass} \,:\, \aloc' \in \locs'\} \,\cup\,
  \{\tuple{\aloc', [j]_\sim} \,:\, \aloc' \in \locs'_\downarrow\}
  \,\subseteq\, \confs^\dag_{j + 1}\]
\item[(iv)]
for all $k \in \{i, \ldots, j\}$,
for some $\locs^k, \locs^k_\downarrow \subseteq \locs$ such that
$\locs^k, \locs^k_\downarrow \,\models\,
 \delta(\aloc_\aformulater, \adataword(k), \uparrow)$
if $\aclass = [k]_\sim$, or such that
$\locs^k, \locs^k_\downarrow \,\models\,
 \delta(\aloc_\aformulater, \adataword(k), \nuparrow)$
if $\aclass \neq [k]_\sim$,
we have
\[\{\tuple{\aloc', \aclass} \,:\, \aloc' \in \locs^k\} \,\cup\,
  \{\tuple{\aloc', [k]_\sim} \,:\, \aloc' \in \locs^k_\downarrow\}
  \,\subseteq\, \confs^\dag_{k + 1}\]
\end{itemize}
By considering subruns starting with
the sets of configurations in (iii) and (iv),
and the inductive hypothesis, it follows that
$\adataword, j \,\models_\aclass\, \aformulabis$, and
$\adataword, k \,\models_\aclass\, \aformulater$
for all $k \in \{i, \ldots, j\}$, so
$\adataword, i \,\models_\aclass\, \aformulabis \release \aformulater$
as required.
If
$\tuple{\aloc_{\aformulabis \release \aformulater}, \aclass} \in
 \confs^\dag_{j + 1}$
for all $j \geq i$,
the argument is again simpler.

For case $\aformula' = {\downarrow} \aformulabis$,
we have $\adataword, i \,\models_\aclass\, {\downarrow} \aformulabis$
iff $\adataword, i \,\models_{[i]_\sim}\, \aformulabis$.
By the inductive hypothesis, that is iff:
\begin{itemize}
\item[(v)]
for some $\locs^\dag, \locs^\dag_\downarrow \subseteq \locs$ such that
$\locs^\dag, \locs^\dag_\downarrow \,\models\,
 \delta(\aloc_\aformulabis, \adataword(i), \uparrow)$,
$\aregaut_\aformula$ has an infinite run
from position $i + 1$ of $\adataword$, starting with
$\{\tuple{\aloc', [i]_\sim} \,:\,
   \aloc' \in \locs^\dag \cup \locs^\dag_\downarrow\}$.
\end{itemize}
On the other hand, $\aregaut_\aformula$ having an infinite run
from position $i + 1$ of $\adataword$, starting with
\[\{\tuple{\aloc', \aclass} \,:\, \aloc' \in \locs'\} \,\cup\,
  \{\tuple{\aloc', [i]_\sim} \,:\, \aloc' \in \locs'_\downarrow\}\]
for some $\locs', \locs'_\downarrow \subseteq \locs$ such that
$\locs', \locs'_\downarrow \,\models\,
 \delta(\aloc_{{\downarrow} \aformulabis}, \adataword(i), \uparrow)$
if $\aclass = [i]_\sim$, or such that
$\locs', \locs'_\downarrow \,\models\,
 \delta(\aloc_{{\downarrow} \aformulabis}, \adataword(i), \nuparrow)$
if $\aclass \neq [i]_\sim$,
is equivalent to:
\begin{itemize}
\item[(vi)]
for some $\locs', \locs'_\downarrow \subseteq \locs$ such that
$\locs', \locs'_\downarrow \,\models\,
 \delta(\aloc_\aformulabis, \adataword(i), \uparrow)
 [{\downarrow} \aloc' / \aloc' \,:\, \aloc' \in \locs]$,
$\aregaut_\aformula$ has an infinite run
from position $i + 1$ of $\adataword$, starting with
\[\{\tuple{\aloc', \aclass} \,:\, \aloc' \in \locs'\} \,\cup\,
  \{\tuple{\aloc', [i]_\sim} \,:\, \aloc' \in \locs'_\downarrow\}\]
\end{itemize}
It remains to observe that
$\locs', \locs'_\downarrow \,\models\,
 \delta(\aloc_\aformulabis, \adataword(i), \uparrow)
 [{\downarrow} \aloc' / \aloc' \,:\, \aloc' \in \locs]$ iff
$\locs'_\downarrow = \locs^\dag \cup \locs^\dag_\downarrow$ for some
$\locs^\dag, \locs^\dag_\downarrow \,\models\,
 \delta(\aloc_\aformulabis, \adataword(i), \uparrow)$,
so (v) and (vi) are equivalent.
\end{proof}

\subsection{Counter Automata}

We introduce below a class of nondeterministic automata on $\omega$-words
which have $\emptyword$ transitions and $\mathbb{N}$-valued counters.
The set of counters of such an automaton will have structure:
there will be a finite set called the basis of the automaton,
and each counter will be a nonempty subset of the basis.
In the course of a transition, the automaton will be able
either to increment a counter, or to decrement a counter if nonzero,
or to perform a simultaneous nondeterministic transfer
with respect to a mapping $\amap$ from counters to sets of counters.
The latter transfers the value of each counter $c$
to the counters in $\amap(c)$, nondeterministically splitting it.
However, only mappings which satisfy a distributivity constraint
in terms of the structure of the set of counters may be used.

The observation that simultaneous nondeterministic transfers
arising from translating safety 1ARA$_1$ are distributive
(cf.\ the proof of Theorem~\ref{th:RA2IPCANT}),
and that distributivity enables nonemptiness of the counter automata
to be decided in space exponential in basis size
(cf.\ the proof of Theorem~\ref{th:IPCANT}),
are key components of the paper.

We shall only consider automata with no cycles of $\emptyword$ transitions,
and they will recognise safety languages,
so every infinite run will accept some $\omega$-word.

The automata will be faulty in the sense that
their counters may erroneously increase at any time.

Formally, for a finite set $\aset$ and
$C \,\subseteq\, \mathcal{P}(\aset) \setminus \{\emptyset\}$,
let $L(C)$ be the set of all instructions:
\begin{itemize}
\item
$\tuple{\mathtt{inc}, c}$ and $\tuple{\mathtt{dec}, c}$ for $c \in C$;
\item
$\tuple{\mathtt{transf}, \amap}$ for mappings
$\amap: C \rightarrow \mathcal{P}(C)$ which are \emph{distributive} as follows:
whenever $c \in C$, $c \subseteq \bigcup_{i = 1}^k c_i$,
and $c'_i \in \amap(c_i)$ for each $i = 1, \ldots, k$,
there exists $c' \in \amap(c)$ such that
$c' \subseteq \bigcup_{i = 1}^k c'_i$.
\end{itemize}

A \emph{safety powerset counter automaton with
nondeterministic transfers and incrementing errors}
(shortly, \emph{safety IPCANT}) $\acaut$ is a tuple
$\tuple{\aalphabet, \locs, \aloc_I, \aset, C, \delta}$ such that:
\begin{itemize}
\item
$\aalphabet$ is a finite alphabet;
\item
$\locs$ is a finite set of states, and
$\aloc_I$ is the initial state;
\item
$\aset$ is a finite set called the \emph{basis}, and
$C \,\subseteq\, \mathcal{P}(\aset) \setminus \{\emptyset\}$
is the set of counters;
\item
$\delta \subseteq
 \locs \times (\aalphabet \uplus \{\emptyword\}) \times L(C) \times \locs$
is a transition relation which does not contain
a cycle of $\emptyword$ transitions.
\end{itemize}

A configuration of $\acaut$ is a pair $\tuple{\aloc, \acval}$,
where $\aloc \in \locs$ and $\acval$ is a counter valuation,
i.e.\ $\acval: C \rightarrow \mathbb{N}$.
We say that $\tuple{\aloc, \acval}$ has an error-free transition
labelled by $\aword \in \aalphabet \uplus \{\emptyword\}$
and performing $l \in L(C)$ to $\tuple{\aloc', \acval'}$,
and we write
$\tuple{\aloc, \acval} \stackrel{\aword, l}{\longrightarrow}_\surd
 \tuple{\aloc', \acval'}$,
iff $\tuple{\aloc, \aword, l, \aloc'} \in \delta$ and
$\acval'$ can be obtained from $\acval$ by $l$.
The latter is defined as follows:
\begin{itemize}
\item
instructions $\tuple{\mathtt{inc}, c}$ and $\tuple{\mathtt{dec}, c}$
have the standard interpretations,
where $\tuple{\mathtt{dec}, c}$ is firable iff $\acval(c) > 0$;
\item
$\acval'$ can be obtained from $\acval$
by $\tuple{\mathtt{transf}, \amap}$ iff
there exist $K^c_{c'} \geq 0$ for each $c \in C$ and $c' \in \amap(c)$,
such that:
\[\mathrm{for\ each}\ c \in C,\
  \acval(c) = \textstyle{\sum}_{c' \in \amap(c)} K^c_{c'}
  \hspace{2em}
  \mathrm{for\ each}\ c' \in C,\
  \acval'(c') = \textstyle{\sum}_{\amap(c) \ni c'} K^c_{c'}\]
in particular, $\tuple{\mathtt{transf}, \amap}$ is firable iff
$\acval(c) = 0$ whenever $\amap(c) = \emptyset$.
\end{itemize}

For counter valuations $\acval$ and $\acval_\surd$,
we write $\acval \leq \acval_\surd$ iff,
for all $c$, $\acval(c) \leq \acval_\surd(c)$.
To allow transitions of $\acaut$ to contain incrementing errors, we define
$\tuple{\aloc, \acval} \stackrel{\aword, l}{\longrightarrow}
 \tuple{\aloc', \acval'}$
to mean that there exist $\acval_\surd$ and $\acval'_\surd$ with
$\acval \leq \acval_\surd$,
$\tuple{\aloc, \acval_\surd} \stackrel{\aword, l}{\longrightarrow}_\surd
 \tuple{\aloc', \acval'_\surd}$ and
$\acval'_\surd \leq \acval'$.

We say that $\acaut$ accepts an $\omega$-word $\aword$ over $\aalphabet$
iff $\acaut$ has a run
$\tuple{\aloc_0, \acval_0} \stackrel{\aword_0, l_0}{\longrightarrow}
 \tuple{\aloc_1, \acval_1} \stackrel{\aword_1, l_1}{\longrightarrow}
 \cdots$
where $\tuple{\aloc_0, \acval_0}$ is the initial configuration
$\tuple{\aloc_I, \mathbf{0}}$
and $\aword = \aword_0 \aword_1 \ldots$.

\begin{example}
\label{ex:IPCANT}
Given $\asetbis \subseteq \aset$,
let $\amap_\asetbis(c) = \emptyset$ if $c \cap \asetbis \neq \emptyset$,
and $\amap_\asetbis(c) = \{c\}$ otherwise.
Observe that $\amap_\asetbis$ is distributive.
The instruction $\tuple{\mathtt{transf}, \amap_\asetbis}$ is firable
iff each counter which intersects $\asetbis$ is zero,
and it does not change the value of any counter.
Hence, we may write $\tuple{\mathtt{ifz}^\cap, \asetbis}$
instead of $\tuple{\mathtt{transf}, \amap_\asetbis}$.

Suppose $C = \{\{\aelem\} \,:\, \aelem \in \aset\}$,
i.e.\ the set of counters has no structure.
The instruction $\tuple{\mathtt{ifz}^\cap, \asetbis}$ is firable
iff each counter $\{\aelem\}$ for $\aelem \in \asetbis$ is zero.
Observe that every $\amap: C \rightarrow \mathcal{P}(C)$ is distributive.
For instance, given $c \in C$ and nonempty $C' \subseteq C$,
let $\amap_{c, C'}(c) = C'$
and $\amap_{c, C'}(c') = \{c'\}$ for $c' \neq c$.
The instruction $\tuple{\mathtt{transf}, \amap_{c, C'}}$
nondeterministically distributes the value of $c$ to the counters in $C'$.
\end{example}

For $\acaut$ as above, let us say that a transition
$\tuple{\aloc, \acval} \stackrel{\aword, l}{\longrightarrow}
 \tuple{\aloc', \acval'}$
is \emph{lazy} iff either
$\tuple{\aloc, \acval} \stackrel{\aword, l}{\longrightarrow}_\surd
 \tuple{\aloc', \acval'}$,
or $l$ is of the form $\tuple{\mathtt{dec}, c}$,
$\acval(c) = 0$ and $\acval' = \acval$.
Thus, in lazy transitions, only incrementing errors which
enable decrements of counters with value $0$ may occur.
The following straightforward proposition shows that
restricting to lazy transitions does not affect
the languages of safety IPCANTs.

\begin{proposition}
\label{pr:lazy}
Whenever
$\tuple{\aloc, \acval} \stackrel{\aword, l}{\longrightarrow}
 \tuple{\aloc', \acval'}$
is a transition of a safety IPCANT $\acaut$
and $\acval_\dag \leq \acval$,
there exists a lazy transition
$\tuple{\aloc, \acval_\dag} \stackrel{\aword, l}{\longrightarrow}
 \tuple{\aloc', \acval'_\dag}$
of $\acaut$ such that $\acval'_\dag \leq \acval'$.
\end{proposition}

A set $L$ of $\omega$-words over an alphabet $\aalphabet$
is called \emph{safety} \cite{Alpern&Schneider87} iff
it is closed under limits of finite prefixes,
i.e.\ for each $\omega$-word $\aword$,
if for each $i > 0$ there exists $\aword'_i \in L$
such that the $i$-prefixes of $\aword$ and $\aword'_i$ are equal,
then $\aword \in L$.
For each safety IPCANT, the tree of all its lazy runs is finitely branching,
so by simplifying the argument in the proof of Proposition~\ref{pr:safety.RA},
and by Proposition~\ref{pr:lazy}, we obtain:

\begin{proposition}
The language of each safety IPCANT is safety.
\end{proposition}

\section{Upper Bound}
\label{s:upper}

This section contains a two-stage proof that
nonemptiness of safety 1ARA$_1$ is in \textsc{ExpSpace}.
The first theorem below shows that each such automaton $\aregaut$
is translatable to a safety IPCANT $\acaut_\aregaut$
of at most exponential size, but whose basis size is polynomially
(in fact, linearly) bounded.  Nonemptiness is preserved,
since $\acaut_\aregaut$ accepts exactly the string projections
of data $\omega$-words in the language of $\aregaut$.
By the second theorem, nonemptiness of $\acaut_\aregaut$
is decidable in space exponential in its basis size and polynomial
(in fact, polylogarithmic) in its alphabet size and number of states,
so space exponential in the size of $\aregaut$ suffices overall.

We start with a piece of notation and a lemma about IPCANT.
Suppose $C$ is a set of counters over a basis $\aset$.
For counter valuations $\acval_\surd$ and $\acval$,
let us write $\acval_\surd \sqsubseteq \acval$ iff
there exists $\acval_\dag \leq \acval$ which can be obtained from $\acval_\surd$
by performing $\tuple{\mathtt{transf}, c \mapsto \{d \,:\, c \subseteq d\}}$.
The lemma states that $\sqsubseteq$ is downwards compatible with
every simultaneous nondeterministic transfer.

\begin{lemma}
\label{l:sqsse}
Whenever $\acval_\surd \sqsubseteq \acval$
and $\acval'$ is obtainable from $\acval$ by some
$\tuple{\mathtt{transf}, \amap}$ with distributive $\amap$,
there exists $\acval'_\surd$ obtainable from $\acval_\surd$ by
$\tuple{\mathtt{transf}, \amap}$ and such that
$\acval'_\surd \sqsubseteq \acval'$.
\end{lemma}

\begin{proof}
We use the following shorthand:
$\widetilde{\acval} =
 \bigcup_{c \in C} \{\tuple{c, 1}, \ldots, \tuple{c, \acval(c)}\}$.

The assumptions are equivalent to existence of:
an injective $\iota: \widetilde{\acval_\surd} \rightarrow \widetilde{\acval}$
such that $c \subseteq d$ whenever $\iota\tuple{c, i} = \tuple{d, j}$,
and a bijective $\beta: \widetilde{\acval} \rightarrow \widetilde{\acval'}$
such that $\amap(d) \ni d'$ whenever $\beta\tuple{d, j} = \tuple{d', j'}$.

For each $\tuple{c, i} \in \widetilde{\acval_\surd}$,
we have $c \subseteq d$ where $\iota\tuple{c, i} = \tuple{d, j}$,
and $\amap(d) \ni d'$ where $\beta\tuple{d, j} = \tuple{d', j'}$,
so by distributivity of $\amap$,
there exists $c' \in \amap(c)$ such that $c' \subseteq d'$.
Hence, there exist a counter valuation $\acval'_\surd$ and a bijective
$\beta_\surd: \widetilde{\acval_\surd} \rightarrow \widetilde{\acval'_\surd}$
such that $c' \in \amap(c)$ and $c' \subseteq d'$ whenever
$\beta_\surd\tuple{c, i} = \tuple{c', i'}$ and
$(\beta \circ \iota)\tuple{c, i} = \tuple{d', j'}$.
It remains to observe that $\beta \circ \iota \circ \beta_\surd^{-1}$
is an injection from $\widetilde{\acval'_\surd}$ to $\widetilde{\acval'}$.
\end{proof}

\begin{theorem}
\label{th:RA2IPCANT}
Given a safety 1ARA$_1$ $\aregaut$,
a safety IPCANT $\acaut_\aregaut$ is computable in polynomial space,
such that $\acaut_\aregaut$ and $\aregaut$ have the same alphabet,
the basis size of $\acaut_\aregaut$
is linear in the number of states of $\aregaut$, and
$\mathrm{L}(\acaut_\aregaut) =
 \{\mathrm{str}(\adataword) \,:\, \adataword \in \mathrm{L}(\aregaut)\}$.
\end{theorem}

\begin{proof}
The proof is an adaptation of the proof of \cite[Theorem~4.4]{Demri&Lazic09},
where it was shown how to translate in polynomial space
weak 1ARA$_1$ to B\"uchi nondeterministic counter automata
with $\emptyword$ transitions and incrementing errors,
and whose instructions are increments, decrements and
zero tests of individual counters.
We show below essentially that, since $\aregaut$ is safety,
zero tests of individual counters, cycles of $\emptyword$ transitions
and the B\"uchi acceptance condition can be eliminated
using nondeterministic transfers with a suitable basis and set of counters,
resulting in a safety IPCANT.

Let $\aregaut = \tuple{\aalphabet, \locs, \aloc_I, \delta}$.
We first introduce an abstraction which maps
a finite set $\confs$ of configurations of $\aregaut$
at a position $i$ of a data word $\adataword$ over $\aalphabet$
to a triple $\tuple{\aletter, \locs_\uparrow, \sharp}$ such that:
$\aletter = \adataword(i)$,
$\locs_\uparrow$ is the set of all states
that occur in $\confs$ paired with $[i]_\sim$,
and for each nonempty $\locsbis \subseteq \locs$,
$\sharp(\locsbis)$ is the number of data $\aclass \neq [i]_\sim$
for which $\locsbis$ is the set of all states
that occur in $\confs$ paired with $\aclass$.
Thus, the abstraction records only the letter at position $i$,
and equalities among the datum at position $i$ and
data in configurations in $\confs$.
We then observe that nonemptiness of $\aregaut$ is equivalent to
existence of an infinite sequence of abstract transitions
which starts from a triple of the form
$\tuple{\aletter, \{\aloc_I\}, \mathbf{0}}$.
In other words, searching for
a data $\omega$-word $\adataword$ over $\aalphabet$
and an infinite run of $\aregaut$ on $\adataword$
can be performed one position at a time,
while keeping in memory only the information recorded by the abstraction.

Formally, we define $H_\aregaut$ to be the set of all triples
$\tuple{\aletter, \locs_\uparrow, \sharp}$ for which
$\aletter \in \aalphabet$, $\locs_\uparrow \subseteq \locs$, and
$\sharp: \mathcal{P}(\locs) \setminus \{\emptyset\} \,\rightarrow\, \mathbb{N}$.
For a data word $\adataword$ over $\aalphabet$,
a position $0 \leq i < \length{\adataword}$,
and finite set $\confs$ of configurations, let
$h(\adataword, i, \confs) =
 \tuple{\adataword(i), \locs_\uparrow^{\confs, [i]_\sim},
        \sharp^{\confs, [i]_\sim}}$,
where, for each nonempty $\locsbis \subseteq \locs$:
\[\locs_\uparrow^{\confs, \aclass} =
  \{\aloc \,:\, \tuple{\aloc, \aclass} \in \confs\}
  \hspace{2em}
  \sharp^{\confs, \aclass}(\locsbis) =
  \length{\{\aclass' \neq \aclass \,:\,
            \locs_\uparrow^{\confs, \aclass'} = \locsbis\}}\]

To obtain a successor of a member of $H_\aregaut$,
for each configuration that it represents,
sets of states which satisfy the appropriate
positive Boolean formula in $\aregaut$ are chosen,
and then two cases are distinguished:
either the datum at the next position occurs in the next set of configurations,
or not.
Thus, we write
$\tuple{\aletter, \locs_\uparrow, \sharp} \rightarrow
 \tuple{\aletter', \locs'_\uparrow, \sharp'}$
iff, for each $\aloc \in \locs_\uparrow$,
there exist
$\locs^\aloc, \locs^\aloc_\downarrow \,\models\,
 \delta(\aloc, \aletter, \uparrow)$,
and for each nonempty $\locsbis \subseteq \locs$,
$j \in \{1, \ldots, \sharp(\locsbis)\}$ and $\aloc \in \locsbis$,
there exist
$\locs^{\locsbis, j, \aloc}, \locs^{\locsbis, j, \aloc}_\downarrow \,\models\,
 \delta(\aloc, \aletter, \nuparrow)$,
such that:
\begin{itemize}
\item
either $\sharp' = \sharp^\dag[\locs'_\uparrow \mapsto
                              \sharp^\dag(\locs'_\uparrow) - 1]$,
\item
or $\locs'_\uparrow = \emptyset$ and $\sharp' = \sharp^\dag$,
\end{itemize}
where, for each nonempty $\locsbis' \subseteq \locs$,
$\sharp^\dag(\locsbis')$ is defined as
\[\begin{array}{c}
  \length{\{\tuple{\locsbis, j} \,:\,
            \textstyle{\bigcup}_{\aloc \in \locsbis} \locs^{\locsbis, j, \aloc}
            = \locsbis'\}} \; + \\
  \left\{\begin{array}{ll}
  1, & \mathrm{if\ }
  \bigcup_{\aloc \in \locs_\uparrow} \locs^\aloc \,\cup\,
  \bigcup_{\aloc \in \locs_\uparrow} \locs^\aloc_\downarrow \,\cup\,
  \bigcup_{\locsbis, j}
    \bigcup_{\aloc \in \locsbis} \locs^{\locsbis, j, \aloc}_\downarrow
  = \locsbis' \\
  0, & \mathrm{otherwise}
  \end{array}\right.
\end{array}\]

We claim the following correspondence between
infinite sequences of transitions in $H_\aregaut$ from initial triples and
infinite runs of $\aregaut$ from initial configurations:
\begin{describe}{(*)}
\item[(*)]
$\tuple{\aletter_0, \locs^0_\uparrow, \sharp_0} \rightarrow
 \tuple{\aletter_1, \locs^1_\uparrow, \sharp_1} \rightarrow
 \cdots$
is an infinite sequence of transitions in $H_\aregaut$ such that
$\locs^0_\uparrow = \{\aloc_I\}$ and $\sharp_0 = \mathbf{0}$
iff $\aregaut$ has an infinite run
$\confs_0 \stackrel{\adataword, 0}{\longrightarrow}
 \confs_1 \stackrel{\adataword, 1}{\longrightarrow}
 \cdots$
on a data $\omega$-word $\adataword$ over $\aalphabet$ such that
$\confs_0 = \{\tuple{\aloc_I, [0]_\sim}\}$ and
$\tuple{\aletter_i, \locs^i_\uparrow, \sharp_i} =
 h(\adataword, i, \confs_i)$
for each $i \in \mathbb{N}$.
\end{describe}
One direction is straightforward, since
$h(\adataword, 0, \{\tuple{\aloc_I, [0]_\sim}\}) =
 \tuple{\adataword(0), \{\aloc_I\}, \mathbf{0}}$, and
$\confs \stackrel{\adataword, i}{\longrightarrow} \confs'$ implies
$h(\adataword, i, \confs) \rightarrow h(\adataword, i + 1, \confs')$.
For the other direction, suppose
$\tuple{\aletter_0, \locs^0_\uparrow, \sharp_0} \rightarrow
 \tuple{\aletter_1, \locs^1_\uparrow, \sharp_1} \rightarrow
 \cdots$
is an infinite sequence of transitions in $H_\aregaut$,
$\locs^0_\uparrow = \{\aloc_I\}$ and $\sharp_0 = \mathbf{0}$.
For each $i \in \mathbb{N}$,
let $\adataword_i$ be a data word over $\aalphabet$ of length $i + 1$
and $\confs_i$ be a set of configurations for $\adataword_i$ with
$\tuple{\aletter_i, \locs^i_\uparrow, \sharp_i} =
 h(\adataword_i, i, \confs_i)$,
chosen as follows:
\begin{itemize}
\item
We take $\mathrm{str}(\adataword_0) = \aletter_0$,
$\sim^{\adataword_0} = \{\tuple{0, 0}\}$, and
$\confs_0 = \{\tuple{\aloc_I, \{0\}}\}$.
\item
Given $\adataword_i$ and $\confs_i$,
we choose $\adataword_{i + 1}$ and $\confs_{i + 1}$ for which
$\adataword_i$ is the $(i + 1)$-prefix of $\adataword_{i + 1}$,
$\tuple{\aletter_{i + 1}, \locs^{i + 1}_\uparrow, \sharp_{i + 1}} =
 h(\adataword_{i + 1}, i + 1, \confs_{i + 1})$, and
$\confs_i \stackrel{\adataword_i, i}{\longrightarrow} \confs_{i + 1}$.
\end{itemize}
Now, let $\adataword^\dag$ be the limit of the $\adataword_i$,
i.e.\ such that for each $i \in \mathbb{N}$,
$\adataword_i$ is the $(i + 1)$-prefix of $\adataword^\dag$.
For each $i \in \mathbb{N}$, let $\confs^\dag_i$ be the unique
set of configurations for $\adataword^\dag$ that satisfies
\[\confs_i =
  \{\tuple{\aloc, \aclass \cap \{0, \ldots, i\}} \,:\,
    \tuple{\aloc, \aclass} \in \confs^\dag_i\}\]
Observe that $\length{\confs^\dag_i} = \length{\confs_i}$,
so $\confs^\dag_i$ is finite.
Moreover,
$h(\adataword^\dag, i, \confs^\dag_i) =
 h(\adataword_i, i, \confs_i)$, so
$h(\adataword^\dag, i, \confs^\dag_i) =
 \tuple{\aletter_i, \locs^i_\uparrow, \sharp_i}$.
Finally, since
$\confs_i \stackrel{\adataword_i, i}{\longrightarrow}
 \confs_{i + 1}$, we have
$\confs^\dag_i \stackrel{\adataword^\dag, i}{\longrightarrow}
 \confs^\dag_{i + 1}$.

The nondeterministic procedure below guesses an infinite sequence
$\tuple{\aletter_0, \locs^0_\uparrow, \sharp_0} \rightarrow
 \tuple{\aletter_1, \locs^1_\uparrow, \sharp_1} \rightarrow
 \cdots$
of transitions in $H_\aregaut$ such that
$\locs^0_\uparrow = \{\aloc_I\}$ and $\sharp_0 = \mathbf{0}$
in the following manner:
whenever the main loop has been performed $i$ times
and execution is at the end of step~(2),
$\aletter$, $\locs_\uparrow$ and the counters $c$
store $\aletter_i$, $\locs^i_\uparrow$ and $\sharp_i$ (respectively),
and all the counters $d$ have value $0$.
In the notation of the definition above of transitions in $H_\aregaut$,
each $d(\locsbis', \locsbis'_\downarrow)$ is used to count
the number of pairs $\tuple{\locsbis, j}$ such that
$\bigcup_{\aloc \in \locsbis} \locs^{\locsbis, j, \aloc} =
 \locsbis'$ and
$\bigcup_{\aloc \in \locsbis} \locs^{\locsbis, j, \aloc}_\downarrow =
 \locsbis'_\downarrow$.
If one or more choices in steps (3) or (4) are not possible,
the procedure blocks.
\begin{itemize}
\item[(0)]
Set $c(\locsbis) := 0$
for each nonempty $\locsbis \subseteq \locs$,
and $d(\locsbis, \locsbis_\downarrow) := 0$
for each $\locsbis, \locsbis_\downarrow \subseteq \locs$.
\item[(1)]
Set $\locs_\uparrow := \{\aloc_I\}$.
\item[(2)]
Choose $\aletter \in \aalphabet$.
\item[(3)]
While $c(\locsbis) > 0$
for some nonempty $\locsbis \subseteq \locs$, do:
\begin{itemize}
\item
decrement $c(\locsbis)$;
\item
for each $\aloc \in \locsbis$, choose
$\locs^\aloc, \locs^\aloc_\downarrow \,\models\,
 \delta(\aloc, \aletter, \nuparrow)$;
\item
increment
$d(\bigcup_{\aloc \in \locsbis} \locs^\aloc,
   \bigcup_{\aloc \in \locsbis} \locs^\aloc_\downarrow)$.
\end{itemize}
\item[(4)]
For each $\aloc \in \locs_\uparrow$, choose
$\locs^\aloc, \locs^\aloc_\downarrow \,\models\,
 \delta(\aloc, \aletter, \uparrow)$.
\item[(5)]
Increment
$c(\bigcup_{\aloc \in \locs_\uparrow} \locs^\aloc \,\cup\,
   \bigcup_{\aloc \in \locs_\uparrow} \locs^\aloc_\downarrow \,\cup\,
   \bigcup_{d(\locsbis, \locsbis_\downarrow) > 0}
     \locsbis_\downarrow)$.
\item[(6)]
While $d(\locsbis, \locsbis_\downarrow) > 0$
for some $\locsbis, \locsbis_\downarrow \subseteq \locs$,
decrement $d(\locsbis, \locsbis_\downarrow)$,
and increment $c(\locsbis)$ if $\locsbis$ is nonempty.
\item[(7)]
Either choose nonempty $\locs_\uparrow$ with $c(\locs_\uparrow) > 0$
and decrement $c(\locs_\uparrow)$,
or $\locs_\uparrow := \emptyset$.
\item[(8)]
Repeat from (2).
\end{itemize}

By (*), we have that the procedure has an infinite execution such that
the letters chosen in step~(2) are $\aletter_0, \aletter_1, \ldots$
iff $\aregaut$ accepts a data $\omega$-word $\adataword$ such that
$\aletter_i = \adataword(i)$ for each $i \in \mathbb{N}$.
Therefore, in the remainder of the proof, we show that
the procedure is implementable by a safety IPCANT $\acaut_\aregaut$
which is computable in polynomial space and
whose basis size is linear in $\length{\locs}$.

For $\locsbis, \locsbis_\downarrow \subseteq \locs$, let
\[\overline{\locsbis} =
  \{\overline{*}\} \,\cup\,
  \{\overline{\aloc} \,:\, \aloc \in \locsbis\}
  \hspace{2em}
  \ooverline{\locsbis, \locsbis_\downarrow} =
  \{\ooverline{*}\} \,\cup\,
  \{\ooverline{\aloc} \,:\, \aloc \in \locsbis\} \,\cup\,
  \{\ooverline{\aloc}_\downarrow \,:\, \aloc \in \locsbis_\downarrow\}\]
We define the basis of $\acaut_\aregaut$ as
$\overline{\locs} \cup \ooverline{\locs, \locs}$
(where we assume disjointness),
and the counters of $\acaut_\aregaut$ are:
$\overline{\locsbis}$ for each $\locsbis \subseteq \locs$,
and $\ooverline{\locsbis, \locsbis_\downarrow}$
for each $\locsbis, \locsbis_\downarrow \subseteq \locs$.
The set of counters of $\acaut_\aregaut$ is thus essentially
$\mathcal{P}(\locs) \cup \mathcal{P}(\locs)^2$.
Note that, compared to the procedure above,
$\acaut_\aregaut$ has the extra counter $\overline{\emptyset}$.

The states of $\acaut_\aregaut$ are used for control,
and for storing the letters from $\aalphabet$ as well as
the elements and subsets of $\locs$.
Step~(0) is implemented by default,
and steps (1), (2), (4) and (8) are straightforward.

Step~(3) can be performed by
a single simultaneous nondeterministic transfer,
with the mapping
\[\begin{array}{c}
  \{\overline{\locsbis} \,\mapsto\,
    \{\ooverline{\textstyle{\bigcup}_{\aloc \in \locsbis} \locs^\aloc,
        \textstyle{\bigcup}_{\aloc \in \locsbis} \locs^\aloc_\downarrow} \,:\,
      \fforall{\aloc \in \locsbis}
              {\locs^\aloc, \locs^\aloc_\downarrow \,\models\,
               \delta(\aloc, \aletter, \nuparrow)}\}, \\
    \ooverline{\locsbis, \locsbis_\downarrow} \mapsto
    \{\ooverline{\locsbis, \locsbis_\downarrow}\}
    \::\: \locsbis, \locsbis_\downarrow \subseteq \locs\}
\end{array}\]
whose distributivity is a key component of the paper.
To show that it holds, suppose
$\overline{\locsbis} \subseteq
 \bigcup_{i = 1}^k \overline{\locsbis^i}$, and
$\locs^{i, \aloc}, \locs^{i, \aloc}_\downarrow \,\models\,
 \delta(\aloc, \aletter, \nuparrow)$
for each $i \in \{1, \ldots, k\}$ and $\aloc \in \locsbis^i$.
Given $\aloc \in \locsbis$,
let $i_\aloc$ be such that $\aloc \in \locsbis^{i_\aloc}$.
We then have, as required:
\[\ooverline{\textstyle{\bigcup}_{\aloc \in \locsbis} \locs^{i_\aloc, \aloc},
    \textstyle{\bigcup}_{\aloc \in \locsbis}
      \locs^{i_\aloc, \aloc}_\downarrow} \subseteq
  \textstyle{\bigcup}_{i = 1}^k
    \ooverline{\textstyle{\bigcup}_{\aloc \in \locsbis^i} \locs^{i, \aloc},
      \textstyle{\bigcup}_{\aloc \in \locsbis^i} \locs^{i, \aloc}_\downarrow}\]

The following is an implementation of step~(5):
\begin{itemize}
\item
Set $\locsbis' := \bigcup_{\aloc \in \locs_\uparrow} \locs^\aloc \,\cup\,
                  \bigcup_{\aloc \in \locs_\uparrow} \locs^\aloc_\downarrow$.
\item
For each $\aloc \in \locs$,
either perform the transfer that
verifies that each $\ooverline{\locsbis, \locsbis_\downarrow}$
with $\aloc \in \locsbis_\downarrow$ is zero (cf.\ Example~\ref{ex:IPCANT}),
or choose $\locsbis, \locsbis_\downarrow \subseteq \locs$
with $\aloc \in \locsbis_\downarrow$,
decrement $\ooverline{\locsbis, \locsbis_\downarrow}$,
increment $\ooverline{\locsbis, \locsbis_\downarrow}$
and set $\locsbis' := \locsbis' \cup \{\aloc\}$.
\item
Increment $\overline{\locsbis'}$.
\end{itemize}

For step~(6), we use the transfer with the mapping
\[\{\overline{\locsbis} \mapsto \{\overline{\locsbis}\},
    \ooverline{\locsbis, \locsbis_\downarrow} \mapsto \{\overline{\locsbis}\}
    \,:\, \locsbis, \locsbis_\downarrow \subseteq \locs\}\]
which is distributive since
$\ooverline{\locsbis, \locsbis_\downarrow} \subseteq
 \bigcup_{i = 1}^k \ooverline{\locsbis^i, \locsbis_\downarrow^i}$
implies
$\overline{\locsbis} \subseteq
 \bigcup_{i = 1}^k \overline{\locsbis^i}$.

Finally, in step~(7), if $\locs_\uparrow := \emptyset$ is performed,
then either $\overline{\emptyset}$ is decremented or not.

Observe therefore that the auxiliary counter $\overline{\emptyset}$ is
transferred to $\ooverline{\emptyset, \emptyset}$ in step~(3),
that $\ooverline{\emptyset, \emptyset}$ is
transferred to $\overline{\emptyset}$ in step~(6),
and that those two counters do not affect anything else.

In step~(2), $\acaut_\aregaut$ performs an $\aletter$ transition,
and all other transitions are $\emptyword$.
However, the only cycle in the transition graph of $\acaut_\aregaut$
corresponds to the loop (2)--(8), so the requirement of
no cycles of $\emptyword$ transitions is met.

The only nontrivial aspect of computing $\acaut_\aregaut$
in space polynomial in the size of $\aregaut$ is
the implementation of step~(3).
However, for each $\locsbis \subseteq \locs$, the set
\[\{\ooverline{\textstyle{\bigcup}_{\aloc \in \locsbis} \locs^\aloc,
      \textstyle{\bigcup}_{\aloc \in \locsbis} \locs^\aloc_\downarrow} \,:\,
    \fforall{\aloc \in \locsbis}
            {\locs^\aloc, \locs^\aloc_\downarrow \,\models\,
             \delta(\aloc, \aletter, \nuparrow)}\}\]
can be output by iterating over all mappings
$\aloc \mapsto \tuple{\locs^\aloc, \locs^\aloc_\downarrow}$
from $\locsbis$ to $\mathcal{P}(\locs)^2$.
Each such mapping can be stored in space $2 \length{\locs}^2$,
and deciding
$\locs^\aloc, \locs^\aloc_\downarrow \,\models\,
 \delta(\aloc, \aletter, \nuparrow)$
amounts to evaluating a propositional formula.

It remains to show that incrementing errors cannot cause
$\acaut_\aregaut$ to accept an $\omega$-word $\aletter_0 \aletter_1 \ldots$
which it does not accept without incrementing errors.
Informally, that is the case because
incrementing errors in runs of $\acaut_\aregaut$ amount to
introductions of spurious threads into corresponding runs of $\aregaut$,
which can only make acceptance harder.

Suppose $\acaut_\aregaut$ accepts
an $\omega$-word $\aletter_0 \aletter_1 \ldots$,
i.e.\ the implementation of the procedure above has an infinite execution $E$
which may contain incrementing errors and
which chooses in step~(2) the letters $\aletter_0, \aletter_1, \ldots$.
Below, we define an error-free infinite execution $E_\surd$ such that
the letters chosen in step~(2) are also $\aletter_0, \aletter_1, \ldots$,
and we show by induction that the following are satisfied before each step:
\begin{itemize}
\item[(i)]
$\acval_\surd \sqsubseteq \acval$ (cf.\ Lemma~\ref{l:sqsse}),
where $\acval$ and $\acval_\surd$ are the current counter valuations
in $E$ and $E_\surd$ (respectively);
\item[(ii)]
$\locs_\uparrow^\surd \subseteq \locs_\uparrow$,
if $\locs_\uparrow$ and $\locs_\uparrow^\surd$ are defined, where they are
the current values of the variable in $E$ and $E_\surd$ (respectively).
\end{itemize}
Initially, we have that $\acval$ and $\acval_\surd$ equal $\mathbf{0}$,
and that $\locs_\uparrow$ and $\locs_\uparrow^\surd$ are undefined,
so the inductive base is trivial.
We also have that
$\acval_\surd \sqsubseteq \acval$ and $\acval \leq \acval'$
imply $\acval_\surd \sqsubseteq \acval'$,
i.e.\ the $\sqsubseteq$ relation is preserved by
incrementing errors in the second argument.
\begin{description}
\item[Steps (1) and (2)]
$E_\surd$ performs the same transitions as $E$.
\item[Steps (3) and (6)]
$E_\surd$ performs the transfers as in Lemma~\ref{l:sqsse}.
\item[Step~(4)]
For each $\aloc \in \locs_\uparrow^\surd \subseteq \locs_\uparrow$,
the same $\locs^\aloc$ and $\locs^\aloc_\downarrow$
are chosen in $E_\surd$ as in $E$.
\item[Step~(5)]
For each $\aloc \in \locs$, if there exist
$\locsbis^\surd$ and $\locsbis_\downarrow^\surd \ni \aloc$ such that
$\acval_\surd(\ooverline{\locsbis^\surd, \locsbis_\downarrow^\surd}) > 0$,
we have by (i) that there exist
$\locsbis$ and $\locsbis_\downarrow \ni \aloc$ such that
$\acval(\ooverline{\locsbis, \locsbis_\downarrow}) > 0$.
It follows that $\locsbis'_\surd \subseteq \locsbis'$,
where $\locsbis'$ is the value of the variable
after the implementation of step~(5) is executed in $E$,
and $\locsbis'_\surd$ is the value
after the unique error-free execution in $E_\surd$.
Hence, (i) is preserved.
\item[Step~(7)]
Let $\iota: \widetilde{\acval_\surd} \rightarrow \widetilde{\acval}$
be an injection (cf.\ the proof of Lemma~\ref{l:sqsse}),
and $\locs_\uparrow$ be the value chosen in $E$.
If $\overline{\locs_\uparrow}$ is decremented and
$\iota\tuple{\overline{\locs_\uparrow^\surd}, i} =
 \tuple{\overline{\locs_\uparrow}, j}$
for some $\locs_\uparrow^\surd$, $i$ and $j$
(in particular, $\locs_\uparrow^\surd \subseteq \locs_\uparrow$),
then choose such $\locs_\uparrow^\surd$ in $E_\surd$ and
decrement $\overline{\locs_\uparrow^\surd}$.
Otherwise, choose $\emptyset$ in $E_\surd$ without decrementing.
\end{description}
That completes the definition of $E_\surd$ and the proof.
\end{proof}

\begin{theorem}
\label{th:IPCANT}
Nonemptiness of safety IPCANT
is decidable in space exponential in basis size
and polylogarithmic in alphabet size and number of locations.
\end{theorem}

\begin{proof}
Suppose $\acaut = \tuple{\aalphabet, \locs, \aloc_I, \aset, C, \delta}$
is a safety IPCANT.
By Proposition~\ref{pr:lazy}, $\acaut$ is nonempty iff
it has an infinite sequence of lazy transitions from the initial configuration.

We define positive integers $\alpha_i$ and $U_i$
for $i = 0, \ldots, \length{\aset}$ as follows:
\[
\alpha_0 = \length{\locs}
\hspace{2em}
U_0 = 1
\hspace{2em}
\alpha_{i + 1} = 2 (\length{\aset} - i) \alpha_i U_i^{\length{C}}
\hspace{2em}
U_{i + 1} = 3 \alpha_i U_i^{\length{C}}
\]
Let $m = 2 \alpha_{\length{\aset}} U_{\length{\aset}}^{\length{C}}$.
We shall show:
\begin{describe}{(I)}
\item[(I)]
If $\acaut$ has a sequence of $m - 1$
lazy transitions from the initial configuration,
then it has an infinite sequence.
\end{describe}
Therefore, nonemptiness of $\acaut$ can be decided
nondeterministically by guessing a sequence of $m - 1$
lazy transitions from the initial configuration.
In every such sequence, each transition increases the sum of all counters
by at most $1$, so no counter can exceed $m - 1$.
Since $m < 2^{2^{2 \length{\aset}^2 + \length{\aset}} \log(3 \length{\locs})}$
and $\length{C} < 2^{\length{\aset}}$, a single configuration
can be stored in space $2^{O(\length{\aset}^2)} O(\log \length{\locs})$.
To guess a sequence of length $m - 1$, it suffices to store
at most two configurations, the number of transitions guessed so far,
and a fixed number of variables bounded by
$\length{\acaut} =
 2^{2^{O(\length{\aset})}} O(\length{\aalphabet} \cdot \length{\locs})$
for indexing the transition relation of $\acaut$.
Hence, nonemptiness of $\acaut$ is decidable nondeterministically in space
$2^{O(\length{\aset}^2)} O(\log(\length{\aalphabet} \cdot \length{\locs}))$,
so by Savitch's Theorem, there is a deterministic algorithm of space complexity
$2^{O(\length{\aset}^2)} O(\log(\length{\aalphabet} \cdot \length{\locs})^2)$.

To show (I), suppose $\acaut$ has a sequence of lazy transitions
$S =
\tuple{\aloc_1, \acval_1}
\stackrel{\aword_1, l_1}{\longrightarrow} \cdots
\stackrel{\aword_{m - 1}, l_{m - 1}}{\longrightarrow}
\tuple{\aloc_m, \acval_m}$
from the initial configuration, but no infinite sequence.
By careful repeated uses of the pigeonhole principle and
the distributivity of simultaneous nondeterministic transfers,
we shall obtain the contradiction that
$S$ must contain two equal configurations.
To start with, some state must occur among $\aloc_1, \ldots, \aloc_m$
at least $m / \length{\locs}$ times,
so let $\aloc \in \locs$ and $J_0 \subseteq \{1, \ldots, m\}$
be such that $\length{J_0} = m / \alpha_0 U_0^{\length{C}}$
and $\aloc_j = \aloc$ for each $j \in J_0$.
We claim:
\begin{describe}{(II)}
\item[(II)]
There exist an enumeration
$\aelem_1, \ldots, \aelem_{\length{\aset}}$ of $\aset$,
and for $i = 1, \ldots, \length{\aset}$,
mappings $u_i: C_i \rightarrow \{0, \ldots, U_i - 1\}$ where
$C_i = \{c \in C \::\: \aelem_i \in c \,\wedge\,
           \aelem_1, \ldots, \aelem_{i - 1} \notin c\}$,
and subsets $J_i$ of $\{1, \ldots, m\}$
of size $m / \alpha_i U_i^{\length{C}}$,
such that the following property holds for each $0 \leq i \leq \length{\aset}$:
for all $j \in J_i$, we have that $\aloc_j = \aloc$
and that for all $1 \leq i' \leq i$ and $c \in C_{i'}$,
$\acval_j(c) = u_{i'}(c)$.
\end{describe}

We establish (II) by proving the property inductively on $i$
and simultaneously picking $\aelem_i$, $u_i$ and $J_i$.
The case $i = 0$ is trivial.
Assume that $0 \leq i < \length{\aset}$
and that $\aelem_{i'}$, $u_{i'}$ and $J_{i'}$ for $i' = 1, \ldots, i$
have been picked so that the property holds for $i$.
Let us call a subsequence of $S$ an \emph{$i$-subsequence} iff
there exist consecutive $j, j' \in J_i$
(i.e.\ where there is no $j'' \in J_i$ with $j < j'' < j'$)
such that the subsequence
begins at $\tuple{\aloc_j, \acval_j}$ and
ends at $\tuple{\aloc_{j'}, \acval_{j'}}$.
Let $J'_i \subseteq J_i$ consist of the beginning positions of the
$\length{J_i} / 2 = m / 2 \alpha_i U_i^{\length{C}}$ shortest $i$-subsequences.
The length of the longest of those $i$-subsequences
must be at most $2 \alpha_i U_i^{\length{C}}$,
since otherwise there would be at least $\length{J_i} / 2$
$i$-subsequences of length more than $m / (\length{J_i} / 2)$.
Let
$S^\dag =
\tuple{\aloc_j, \acval_j}
\stackrel{\aword_j, l_j}{\longrightarrow} \cdots
\stackrel{\aword_{j' - 1}, l_{j' - 1}}{\longrightarrow}
\tuple{\aloc_{j'}, \acval_{j'}}$
be an $i$-subsequence with $j \in J'_i$.
We have $j' - j \leq 2 \alpha_i U_i^{\length{C}}$,
$\aloc_j = \aloc_{j'} = \aloc$,
and for all $1 \leq i' \leq i$ and $c \in C_{i'}$,
$\acval_j(c) = \acval_{j'}(c) = u_{i'}(c)$.
Recalling that $u_{i'}: C_{i'} \rightarrow \{0, \ldots, U_{i'} - 1\}$, we obtain
$\sum_{i' = 1}^i \sum_{c \in C_{i'}} \acval_{j'}(c) \leq
 \sum_{i' = 1}^i \length{C_{i'}} U_{i'}$.

To make progress, we prove:
\begin{describe}{(III)}
\item[(III)]
There exists $\aelem'_j \neq \aelem_1, \ldots, \aelem_i$ such that,
for each $c$ with $\aelem'_j \in c$ and $\aelem_1, \ldots, \aelem_i \notin c$,
$\acval_j(c) \leq
 2 \alpha_i U_i^{\length{C}} +
 \sum_{i' = 1}^i \length{C_{i'}} U_{i'}$.
\end{describe}
Suppose the contrary:
for each $\aelem' \neq \aelem_1, \ldots, \aelem_i$,
there exists $c_{\aelem'}$ such that
$\aelem' \in c_{\aelem'}$,
$\aelem_1, \ldots, \aelem_i \notin c_{\aelem'}$, and
$\acval_j(c_{\aelem'}) >
 2 \alpha_i U_i^{\length{C}} +
 \sum_{i' = 1}^i \length{C_{i'}} U_{i'}$.
Let $H$ be a directed acyclic graph on $\{j, \ldots, j'\} \times C$,
defined by letting the successors of $\tuple{k, d}$ be:
\begin{itemize}
\item
$\emptyset$, if $k = j'$;
\item
$\{\tuple{k + 1, d'} \,:\, d' \in \amap(d)\}$,
if $l_k$ is of the form $\tuple{\mathtt{transf}, \amap}$;
\item
$\{\tuple{k + 1, d}\}$, otherwise.
\end{itemize}
Now, for $c \in C$ and $k \in \{j, \ldots, j'\}$,
let $H(c, k)$ be the set of all $d$ such that
$\tuple{k, d}$ is reachable in $H$ from $\tuple{j, c}$.
We have
$\sum_{d \in H(c, k)} \acval_k(d) \geq
 \acval_j(c) - (k - j)$
by induction on $k$.
In particular, for each $\aelem' \neq \aelem_1, \ldots, \aelem_i$, we have
$\sum_{d \in H(c_{\aelem'}, j')} \acval_{j'}(d) \geq
 \acval_j(c_{\aelem'}) - (j' - j) >
 \sum_{i' = 1}^i \length{C_{i'}} U_{i'} \geq
 \sum_{i' = 1}^i \sum_{c \in C_{i'}} \acval_{j'}(c)$,
so there is some $d_{\aelem'} \in H(c_{\aelem'}, j')$ such that
$\aelem_1, \ldots, \aelem_i \notin d_{\aelem'}$.
Let $H_{\aelem'}$ be a path in $H$
from $\tuple{j, c_{\aelem'}}$ to $\tuple{j', d_{\aelem'}}$.
For $k \in \{j, \ldots, j'\}$, let $H_{\aelem'}(k)$ denote
the counter at position $k$ in $H_{\aelem'}$.

Consider any $c$ with $\aelem_1, \ldots, \aelem_i \notin c$.
Observe that $c \subseteq \bigcup \{c_{\aelem'} \,:\, \aelem' \in c\}$.
Let $H_c$ be a path in $H$ from $\tuple{j, c}$, obtained as follows.
Assuming that $k \in \{j, \ldots, j' - 1\}$ and
$H_c(k) \subseteq \bigcup \{H_{\aelem'}(k) \,:\, \aelem' \in c\}$:
\begin{itemize}
\item
if $l_k$ is of the form $\tuple{\mathtt{transf}, \amap}$,
by distributivity of $\amap$ and the definition of $H$, we can pick
$H_c(k + 1) \subseteq \bigcup \{H_{\aelem'}(k + 1) \,:\, \aelem' \in c\}$;
\item
otherwise, we have $H_{\aelem'}(k + 1) = H_{\aelem'}(k)$
for each $\aelem' \in c$, and the only possibility is $H_c(k + 1) = H_c(k)$.
\end{itemize}
Since $H_c(j') \subseteq \bigcup \{H_{\aelem'}(j') \,:\, \aelem' \in c\}$,
we conclude that $\aelem_1, \ldots, \aelem_i \notin H_c(j')$.

Using the paths $H_c$, we now show that,
from the final configuration of $S^\dag$,
the instructions in $S^\dag$ can be performed repeatedly
to obtain an infinite sequence of lazy transitions,
which is a contradiction, so (III) holds.
More precisely, since $\acval_j(d) = \acval_{j'}(d)$
for all $1 \leq i' \leq i$ and $d \in C_{i'}$,
and $H_c(j) = c$ for all $c$,
by (IV) below from $\acval_{j'}$ for $k = j, \ldots, j' - 1$,
there exist lazy transitions
$\tuple{\aloc_j, \acval_{j'}}
 \stackrel{\aword_j, l_j}{\longrightarrow} \cdots
 \stackrel{\aword_{j' - 1}, l_{j' - 1}}{\longrightarrow}
 \tuple{\aloc_{j'}, \acval'_{j'}}$
such that $\acval'_{j'}(d) \leq \acval_{j'}(d)$
for all $d \notin \{H_c(j') \,:\, \aelem_1, \ldots, \aelem_i \notin c\}$.
But
$\{H_c(j') \,:\, \aelem_1, \ldots, \aelem_i \notin c\} \subseteq
 \{c \,:\, \aelem_1, \ldots, \aelem_i \notin c\}$,
so (IV) can be applied from $\acval'_{j'}$ for $k = j, \ldots, j' - 1$, etc.
\begin{describe}{(IV)}
\item[(IV)]
Suppose $k \in \{j, \ldots, j' - 1\}$,
and $\acval'_k$ is a counter valuation
such that $\acval'_k(d) \leq \acval_k(d)$
for all $d \notin \{H_c(k) \,:\, \aelem_1, \ldots, \aelem_i \notin c\}$.
There exists a lazy transition
$\tuple{\aloc_k, \acval'_k}
 \stackrel{\aword_k, l_k}{\longrightarrow}
 \tuple{\aloc_{k + 1}, \acval'_{k + 1}}$
such that $\acval'_{k + 1}(d) \leq \acval_{k + 1}(d)$
for all $d \notin \{H_c(k + 1) \,:\, \aelem_1, \ldots, \aelem_i \notin c\}$.
\end{describe}
To show (IV), we distinguish between two cases:
\begin{itemize}
\item
If $l_k$ is of the form $\tuple{\mathtt{transf}, \amap}$,
let $K^d_{d'} \geq 0$ for each $d \in C$ and $d' \in \amap(d)$ satisfy
\[\begin{array}{c}
  \mathrm{for\ each}\ d \in C,\
  \acval_k(d) = \textstyle{\sum}_{d' \in \amap(d)} K^d_{d'} \\
  \mathrm{for\ each}\ d' \in C,\
  \acval_{k + 1}(d') = \textstyle{\sum}_{\amap(d) \ni d'} K^d_{d'}
  \end{array}\]
For $d \in C$ such that $\acval'_k(d) \leq \acval_k(d)$,
pick any ${K'}^d_{d'} \geq 0$ such that
$\acval'_k(d) = \sum_{d' \in \amap(d)} {K'}^d_{d'}$
and ${K'}^d_{d'} \leq K^d_{d'}$ for each $d' \in \amap(d)$.
For $d \in C$ such that $\acval'_k(d) > \acval_k(d)$,
we have $d = H_c(k)$ for some $c$ with $\aelem_1, \ldots, \aelem_i \notin c$,
so we can set ${K'}^d_{d'} = K^d_{d'}$
for all $d' \in \amap(d) \setminus \{H_c(k + 1)\}$,
and ${K'}^d_{H_c(k + 1)} = K^d_{H_c(k + 1)} + \acval'_k(d) - \acval_k(d)$.
Now, for each $d' \in C$, let
$\acval'_{k + 1}(d') = \sum_{\amap(d) \ni d'} K^d_{d'}$, so that
$\tuple{\aloc_k, \acval'_k}
 \stackrel{\aword_k, l_k}{\longrightarrow}
 \tuple{\aloc_{k + 1}, \acval'_{k + 1}}$ lazily.
Since ${K'}^d_{d'} > K^d_{d'}$ implies
$d' \in \{H_c(k + 1) \,:\, \aelem_1, \ldots, \aelem_i \notin c\}$,
we have $\acval'_{k + 1}(d') \leq \acval_{k + 1}(d')$
for all $d' \notin \{H_c(k + 1) \,:\, \aelem_1, \ldots, \aelem_i \notin c\}$.
\item
Otherwise, $\acval'_{k + 1}$ is uniquely determined by the lazy transition
$\tuple{\aloc_k, \acval'_k}
 \stackrel{\aword_k, l_k}{\longrightarrow}
 \tuple{\aloc_{k + 1}, \acval'_{k + 1}}$,
and has the required property as $H_c(k + 1) = H_c(k)$ for all $c$.
\end{itemize}

For each $j \in J'_i$,
let $\aelem'_j \neq \aelem_1, \ldots, \aelem_i$ be as in (III).
For each $c$ with $\aelem'_j \in c$ and $\aelem_1, \ldots, \aelem_i \notin c$,
we have $\acval_j(c) < U_{i + 1}$.
Let $\aelem_{i + 1}$ be such that
there exists $J''_i \subseteq J'_i$ of size
$\length{J'_i} / (\length{\aset} - i) = m / \alpha_{i + 1}$
with $\aelem_{i + 1} = \aelem'_j$ for all $j \in J''_i$.
Thus, for all $j \in J''_i$ and $c \in C_{i + 1}$,
we have $\acval_j(c) < U_{i + 1}$.
Then let $u_{i + 1}: C_{i + 1} \rightarrow \{0, \ldots, U_{i + 1} - 1\}$
be such that there exists $J_{i + 1} \subseteq J''_i$
of size $m / \alpha_{i + 1} U_{i + 1}^{\length{C}}$
with $\acval_j(c) = u_{i + 1}(c)$
for all $j \in J_{i + 1}$ and $c \in C_{i + 1}$.
That completes the inductive proof of (II).

Since $m = 2 \alpha_{\length{\aset}} U_{\length{\aset}}^{\length{C}}$,
we have from (II) that $S$ contains two equal configurations, so $\acaut$ has
an infinite sequence of lazy transitions from the initial configuration.
That is a contradiction, so (I) is shown.
\end{proof}

By Theorems \ref{th:RA2IPCANT}, \ref{th:IPCANT} and \ref{th:LTL2RA},
we obtain:

\begin{corollary}
Safety 1ARA$_1$ nonemptiness and
safety LTL$^\downarrow_1(\nextt, \release)$ satisfiability
are in \textsc{ExpSpace}.
\end{corollary}

\section{Lower Bound}

\begin{theorem}
Safety 1ARA$_1$ nonemptiness and
safety LTL$^\downarrow_1(\nextt, \release)$ satisfiability
are \textsc{ExpSpace}-hard.
\end{theorem}

\begin{proof}
By Theorem~\ref{th:LTL2RA}, it suffices to show \textsc{ExpSpace}-hardness of
satisfiability for safety LTL$^\downarrow_1(\nextt, \release)$.
We shall reduce from the halting problem for
Turing machines with exponentially long tapes.
More precisely, a Turing machine $\amachine$ is a tuple
$\tuple{\aalphabet, \aletter_B, \locs, \aloc_I, \delta}$
such that:
\begin{itemize}
\item
$\aalphabet$ is a finite alphabet, and
$\aletter_B \in \aalphabet$ denotes the blank symbol;
\item
$\locs$ is a finite set of states, and
$\aloc_I \in \locs$ is the initial state;
\item
$\delta:
 \locs \times \aalphabet \rightarrow
 \locs \times \aalphabet \times \{-1, 1\}$
is the transition function.
\end{itemize}

If the size of $\amachine$ is $n$,
we consider its computation on a tape of length $2^n$.
More formally, a configuration of $\amachine$ is
of the form $\tuple{\aloc, i, \aword}$ where
$\aloc \in \locs$ is the machine state,
$0 \leq i < 2^n$ is the head position, and
$\aword \in \aalphabet^{2^n}$ is the tape contents.
The initial configuration is $\tuple{\aloc_I, 0, \aletter_B^{2^n}}$.
A configuration $\tuple{\aloc, i, \aword}$ has a transition iff
$0 \leq i + o < 2^n$ where
$\tuple{\aloc', \aletter, o} = \delta(\aloc, \aword(i))$.
In that case, we write
$\tuple{\aloc, i, \aword} \rightarrow
 \tuple{\aloc', i + o, \aword[i \mapsto \aletter]}$.
Since $\amachine$ can halt by
requesting to move the head off an edge of the tape,
it does not need to have a special halting state.

The following problem is \textsc{ExpSpace}-complete:
given $\amachine = \tuple{\aalphabet, \aletter_B, \locs, \aloc_I, \delta}$
of size $n$, is the computation from the initial configuration
with tape length $2^n$ infinite?
(To reduce in polynomial time from the same problem with tape length $2^{n^k}$,
extend the machine by unreachable states until it is of size $n^k$.)
We shall show that a sentence $\aformula_\amachine$
of safety LTL$^\downarrow_1(\nextt, \release)$
is computable in space logarithmic in $n$,
such that the answer to the decision problem is `yes'
iff $\aformula_\amachine$ is satisfiable.

Let
$\widehat{\aalphabet} =
 \{\widehat{\aletter} \,:\, \aletter \in \aalphabet\}$.
The alphabet of $\aformula_\amachine$ is
$\widetilde{\aalphabet} =
 \locs \,\uplus\,
 \{0_d, 1_d \,:\, d \in \{1, \ldots, n\}\} \,\uplus\,
 \aalphabet \,\uplus\, \widehat{\aalphabet}$.
To encode a tape cell, we write its position in binary followed by its contents.
A configuration $\tuple{\aloc, i, \aword}$ is then encoded by the word below,
where $\widehat{\aalphabet}$ is used to mark the contents at head position.
Let $\aword(i, i) = \widehat{\aword(i)}$,
and $\aword(j, i) = \aword(j)$ for $j \neq i$.
\[\aloc \,
  0_1 \,\cdots\, 0_{n - 1} \, 0_n \, \aword(0, i) \,
  0_1 \,\cdots\, 0_{n - 1} \, 1_n \, \aword(1, i) \,\cdots\,
  1_1 \,\cdots\, 1_{n - 1} \, 1_n \, \aword(2^n - 1, i)\]

The computation of $\amachine$
from the initial configuration with tape length $2^n$ is infinite iff
there exists a data $\omega$-word $\adataword$ over $\widetilde{\aalphabet}$
such that:
\begin{itemize}
\item[(i)]
$\mathrm{str}(\adataword)$ is a sequence of
encodings of configurations of $\amachine$;
\item[(ii)]
$\mathrm{str}(\adataword)$ begins with the encoding of
the initial configuration $\tuple{\aloc_I, 0, \aletter_B^{2^n}}$;
\item[(iii)]
for every two consecutive encodings in $\mathrm{str}(\adataword)$
of configurations $\tuple{\aloc, i, \aword}$ and $\tuple{\aloc', i', \aword'}$,
we have $\tuple{\aloc, i, \aword} \rightarrow \tuple{\aloc', i', \aword'}$.
\end{itemize}
Hence, it suffices to construct $\aformula_\amachine$ such that
$\adataword$ satisfies $\aformula_\amachine$ iff
(i)--(iii) hold and:
\begin{itemize}
\item[(iv)]
for every encoding in $\adataword$ of a tape cell,
all the letters $b_d$ and $\aword(j, i)$ are in the same class;
\item[(v)]
for every two encodings in $\adataword$
of tape cells with positions $j$ and $j'$
(occuring in one or two configuration encodings),
their classes are the same iff $j = j'$.
\end{itemize}
The purpose of (iv) and (v) is to enable navigation through $\adataword$
for checking (i)--(iii) in $\aformula_\amachine$,
whose size will be only polynomial in $n$.

For (i), we can split it into the following constraints,
each of which is straightforward to express:
\begin{itemize}
\item
the first letter is a state of $\amachine$;
\item
every state of $\amachine$ is succeeded by
$0_1 \,\cdots\, 0_{n - 1} \, 0_n$;
\item
every $b_n$ is succeeded by
an element of $\aalphabet \,\uplus\, \widehat{\aalphabet}$;
\item
for every $b_d$ not succeeded by $1_{d + 1} \,\cdots\, 1_n$,
$b_d$ occurs $n + 1$ positions later
(the next position has the same binary digit $d$);
\item
for every $0_d$ succeeded by $1_{d + 1} \,\cdots\, 1_n$,
$1_d \, 0_{d + 1} \,\cdots\, 0_n$ occurs $n + 1$ positions later
(the next position has the opposite binary digit $d$);
\item
$1_1 \,\cdots\, 1_{n - 1} \, 1_n$ followed by
an element of $\aalphabet \,\uplus\, \widehat{\aalphabet}$
are succeeded by a state of $\amachine$;
\item
between every two consecutive occurrences of states of $\amachine$,
there is exactly one occurrence of an element of $\widehat{\aalphabet}$.
\end{itemize}

Properties (ii) and (iv) are also straightforward.
Before (iii), let us consider (v),
which is equivalent to the following conjunction:
\begin{itemize}
\item[(v.1)]
for every two encodings of tape cells,
if their classes are the same then their positions are the same;
\item[(v.2)]
for every encoding of a tape cell,
some tape cell in the next configuration encoding has the same class.
\end{itemize}
The more involved is (v.1).
It amounts to requiring that,
for all $d \in \{1, \ldots, n\}$ and $b \in \{0, 1\}$,
it is not the case that
there is an occurrence of $b_d$ and
a subsequent occurrence of $(1 - b)_d$ with the same datum:
\[\textstyle{\bigwedge}_{d = 1}^n \textstyle{\bigwedge}_{b = 0}^1
  \always (\overline{b_d} \vee {\downarrow} \nextt
    \always (\overline{(1 - b)_d} \vee {\nuparrow}))\]
where $\overline{\aletter}$ abbreviates
$\bigvee \{\aletter' \::\: \aletter' \,\in\,
             \widetilde{\aalphabet} \setminus \{\aletter\}\}$.

Property (iii) is now equivalent to asserting that
the following hold for all $\aloc \in \locs$ and $\aletter \in \aalphabet$,
where $\tuple{\aloc', \aletter', o} = \delta(\aloc, \aletter)$:
\begin{itemize}
\item[(iii.1)]
whenever $\aloc$ occurs with $\widehat{\aletter}$
in the same configuration encoding,
the next occurrence of a state of $\amachine$ is $\aloc'$;
\item[(iii.2)]
for every occurrence of some $\aletterbis \in \aalphabet$
in a configuration encoding which contains $\aloc$ and $\widehat{\aletter}$,
the next occurrence in the same class of
an element of $\aalphabet \,\uplus\, \widehat{\aalphabet}$
is an occurrence of $\aletterbis$ or $\widehat{\aletterbis}$;
\item[(iii.3)]
for every occurrence of $\widehat{\aletter}$
in a configuration encoding containing $\aloc$,
the next occurrence in the same class of
an element of $\aalphabet \,\uplus\, \widehat{\aalphabet}$
is an occurrence of $\aletter'$,
and $n$ positions earlier (if $o = -1$) or later (if $o = 1$)
an element of $\widehat{\aalphabet}$ occurs.
\end{itemize}
The most involved is (iii.3),
and the two cases of $o = -1$ and $o = 1$ are similar.
Letting $\widehat{\aalphabet}$ and $\overline{\widehat{\aalphabet}}$ abbreviate
$\bigvee \{\aletterbis \,:\, \aletterbis \in \widehat{\aalphabet}\}$ and
$\bigvee \{\aletterbis \::\: \aletterbis \,\in\,
             \widetilde{\aalphabet} \setminus \widehat{\aalphabet}\}$
(respectively), (iii.3) with $o = -1$ is expressed by:
\[\always \bigg(\aloc \Rightarrow
    \neg \bigg(\overline{\widehat{\aalphabet}} \until
      \Big(\widehat{\aletter} \wedge {\downarrow} \nextt
        \Big(\overline{\widehat{\aalphabet}} \until
          \big(\widehat{\aalphabet} \wedge \nextt^n
            \neg (\aletter' \wedge {\uparrow})\big)\Big)\Big)\bigg)\bigg)\]
To obtain a sentence of safety LTL$^\downarrow_1(\nextt, \release)$
in the strict sense, we convert to negation normal form:
\[\always \bigg(\overline{\aloc} \vee
    \bigg(\widehat{\aalphabet} \release
      \Big(\overline{\widehat{\aletter}} \vee {\downarrow} \nextt
        \Big(\widehat{\aalphabet} \release
          \big(\overline{\widehat{\aalphabet}} \vee \nextt^n
            (\aletter' \wedge {\uparrow})\big)\Big)\Big)\bigg)\bigg)\]

To output $\widetilde{\aalphabet}$ and $\aformula_\amachine$
given $\amachine$ as above,
a fixed number of counters which are bounded by $n$ suffice.
\end{proof}

\section{Inclusion and Refinement}

Using well-quasi-orderings,
the proofs of Theorems \ref{th:RA2IPCANT} and \ref{th:IPCANT},
and that satisfiability over finite data words
for LTL$^\downarrow_1(\nextt, \sometimes)$
is not primitive recursive \cite[Theorem~5.2]{Demri&Lazic09},
we obtain the result below.

We remark that, in a similar manner, one can show that the following
``model-checking'' problems are decidable and not primitive recursive:
whether the language of a B\"uchi one-way nondeterministic register automaton
(with any number of registers) is included in the language of
a safety 1ARA$_1$ or a safety LTL$^\downarrow_1(\nextt, \release)$ sentence.

\begin{theorem}
The following problems are decidable and not primitive recursive:
\begin{itemize}
\item
inclusion for safety 1ARA$_1$;
\item
refinement for safety LTL$^\downarrow_1(\nextt, \release)$.
\end{itemize}
\end{theorem}

\begin{proof}
By Theorem~\ref{th:LTL2RA}, it suffices to establish
that inclusion for safety 1ARA$_1$ is decidable and
that refinement for safety LTL$^\downarrow_1(\nextt, \release)$
is not primitive recursive.

For the former, suppose
$\aregaut_1 = \tuple{\aalphabet, \locs_1, \aloc_I^1, \delta_1}$ and
$\aregaut_2 = \tuple{\aalphabet, \locs_2, \aloc_I^2, \delta_2}$
are safety 1ARA$_1$, where we need to determine whether
$\mathrm{L}(\aregaut_1) \subseteq \mathrm{L}(\aregaut_2)$.

Let
$\overline{\aregaut_2} =
 \tuple{\aalphabet, \locs_2, \aloc_I^2, \overline{\delta_2}}$
be the dual automaton to $\aregaut_2$,
so that each formula $\overline{\delta_2}(\alocbis, \aletter, ?)$
is the dual to $\delta_2(\alocbis, \aletter, ?)$,
i.e.\ obtained by replacing every $\top$ with $\bot$,
every $\wedge$ with $\vee$, and vice versa.
Let $\mathrm{L}(\overline{\aregaut_2})$ denote the language of
$\overline{\aregaut_2}$ with respect to \emph{co-safety} acceptance:
a data $\omega$-word $\adataword$ over $\aalphabet$
is in $\mathrm{L}(\overline{\aregaut_2})$
iff $\overline{\aregaut_2}$ has a finite run
$\confs_0 \stackrel{\adataword, 0}{\longrightarrow}
 \confs_1 \stackrel{\adataword, 1}{\longrightarrow}
 \cdots \emptyset$
where $\confs_0 = \{\tuple{\aloc_I^2, [0]_\sim}\}$.
Considering $\aregaut_2$ (resp., $\overline{\aregaut_2}$) as
a weak alternating automaton whose every state is of even (resp., odd) parity,
we have by \cite[Theorem~1]{Loding&Thomas00} that
$\mathrm{L}(\overline{\aregaut_2})$ is the complement of
$\mathrm{L}(\aregaut_2)$.

Now, let $\aregaut_\cap$ be the automaton for
the intersection of $\aregaut_1$ and $\overline{\aregaut_2}$,
obtained by adding a new initial state.
More precisely, assuming that $\locs_1$ and $\locs_2$ are disjoint
and do not contain $\aloc_I$, let
$\aregaut_\cap =
 \tuple{\aalphabet, \{\aloc_I\} \cup \locs_1 \cup \locs_2,
        \aloc_I, \delta_\cap}$,
where
\[\delta_\cap =
  \{\tuple{\aloc_I, \aletter, ?} \mapsto
    \delta_1(\aloc_I^1, \aletter, ?) \wedge
    \overline{\delta_2}(\aloc_I^2, \aletter, ?) \,:\,
    \aletter \in \aalphabet, ? \in \{\uparrow, \nuparrow\}\} \:\cup\:
  \delta_1 \:\cup\:
  \overline{\delta_2}\]
The acceptance condition of $\aregaut_\cap$ is inherited from
$\aregaut_1$ and $\overline{\aregaut_2}$:
a data $\omega$-word $\adataword$ over $\aalphabet$
is in $\mathrm{L}(\aregaut_\cap)$
iff $\aregaut_\cap$ has an infinite run
$\confs_0 \stackrel{\adataword, 0}{\longrightarrow}
 \confs_1 \stackrel{\adataword, 1}{\longrightarrow}
 \cdots$
where $\confs_0 = \{\tuple{\aloc_I, [0]_\sim}\}$
and there exists $i$ such that $\confs_i$ contains only states in $\locs_1$.
We then have that
$\mathrm{L}(\aregaut_\cap) =
 \mathrm{L}(\aregaut_1) \cap \mathrm{L}(\overline{\aregaut_2})$,
so $\mathrm{L}(\aregaut_\cap)$ is empty iff
$\mathrm{L}(\aregaut_1) \subseteq \mathrm{L}(\aregaut_2)$.

Let $\acaut_\cap$ be the IPCANT computed from $\aregaut_\cap$
as in the proof of Theorem~\ref{th:RA2IPCANT},
except that the following step is added between steps (6) and (7),
where $\aloc_\emptyset^2$ is a new state
and implementation is similar to that of step~(5):
\begin{itemize}
\item[(6$\frac{1}{2}$)]
If $c(\locsbis) = 0$ for all $\locsbis$ which intersect $\locs_2$,
then pass through $\aloc_\emptyset^2$.
\end{itemize}
We thus have that $\mathrm{L}(\aregaut_\cap)$ is nonempty iff
$\acaut_\cap$ has an infinite run
$\tuple{\aloc_0, \acval_0} \stackrel{\aword_0, l_0}{\longrightarrow}
 \tuple{\aloc_1, \acval_1} \stackrel{\aword_1, l_1}{\longrightarrow}
 \cdots$
where $\tuple{\aloc_0, \acval_0}$ is the initial configuration
and there exists $i$ such that $\aloc_i = \aloc_\emptyset^2$.

We define $\preceq$ to be the following quasi-ordering
on configurations of $\acaut_\cap$:
$\tuple{\aloc, \acval} \preceq \tuple{\aloc', \acval'}$ iff
$\aloc = \aloc'$ and $\acval \leq \acval'$.
By Dickson's Lemma \cite{Dickson13}, $\preceq$ is a \emph{well-quasi-ordering}:
for every infinite sequence $s_0, s_1, \ldots$,
there exist $i < j$ such that $s_i \preceq s_j$.
Now, consider the following procedure:
\begin{itemize}
\item[(i)]
Let $S$ consist of the initial configuration of $\acaut_\cap$.
\item[(ii)]
Let $S'$ be the set of all successors of configurations in $S$
by lazy transitions.
\item[(iii)]
If for all $s' \in S'$ there exists $s \in S$ with $s \preceq s'$, then stop.
Otherwise, set $S$ to $S \cup S'$, and repeat from (ii).
\end{itemize}
Since $\preceq$ is a well-quasi-ordering, the procedure terminates.
Let $S_{\mathrm{last}}$ denote the value of $S$ at the termination.
It is a finite set, and by Proposition~\ref{pr:lazy}, its upward closure
${\Uparrow} S_{\mathrm{last}} =
 \{s' \,:\, \eexists{s \in S_{\mathrm{last}}}{s \preceq s'}\}$
is the set of all configurations which
$\acaut_\cap$ can reach from the initial configuration.

To conclude decidability of inclusion for safety 1ARA$_1$,
it remains to show that we can decide whether ${\Uparrow} S_{\mathrm{last}}$
contains a configuration whose state is $\aloc_\emptyset^2$
and from which $\acaut_\cap$ has an infinite run.
But that is the case iff $S_{\mathrm{last}}$ contains such a configuration,
and for any configuration $\tuple{\aloc, \acval}$,
we have by the proof of Theorem~\ref{th:IPCANT} that
$\acaut_\cap$ has an infinite run from $\tuple{\aloc, \acval}$
iff it has a sequence of $m - 1$ lazy transitions from $\tuple{\aloc, \acval}$,
where $m$ is as computed in that proof.

We now turn to showing that already validity
for safety LTL$^\downarrow_1(\nextt, \release)$
is not primitive recursive.
We reduce (in logarithmic space) from
satisfiability over finite data words
for LTL$^\downarrow_1(\nextt, \sometimes)$,
which is not primitive recursive by \cite[Theorem~5.2]{Demri&Lazic09}.
In negation normal form, the latter logic differs from
safety LTL$^\downarrow_1(\nextt, \release)$ by having
temporal operators $\dnext$, $\sometimes$ and $\always$ instead of $\release$.
Over finite data words, $\nextt$ and its dual $\dnext$ are distinct:
at any final word position and for any $\aformula$,
$\nextt \aformula$ is false whereas $\dnext \aformula$ is true.

Consider the following translation
from formulae of LTL$^\downarrow_1(\nextt, \sometimes)$
in negation normal form with alphabet $\aalphabet$
to formulae of co-safety LTL$^\downarrow_1(\nextt, \release)$
with alphabet $\aalphabet \uplus \{\times\}$.
Only cases where the construct is modified are shown.
\[\begin{array}{rcl@{\hspace{2em}}rcl}
t(\nextt \aformula) & = &
\nextt (t(\aformula) \wedge \bigvee_{\aletter \in \aalphabet} \aletter)
&
t(\sometimes \aformula) & = &
(\bigvee_{\aletter \in \aalphabet} \aletter) \until
(t(\aformula) \wedge \bigvee_{\aletter \in \aalphabet} \aletter)
\\
t(\dnext \aformula) & = &
\nextt (t(\aformula) \vee {\times})
&
t(\always \aformula) & = &
(t(\aformula) \wedge \bigvee_{\aletter \in \aalphabet} \aletter) \until
{\times}
\end{array}\]
Given a sentence $\aformula$, we have that
a data $\omega$-word $\adataword$ over $\aalphabet \uplus \{\times\}$ satisfies
$\aformulabis_\aformula =
 t(\aformula) \wedge (\bigvee_{\aletter \in \aalphabet} \aletter)
 \wedge (\top \until {\times})$
iff there exists $i > 0$ such that the $i$-prefix of $\adataword$
does not contain $\times$ and satisfies $\aformula$,
and $\adataword(i) = {\times}$.
It remains to observe that the dual of $\aformulabis_\aformula$
is a sentence of safety LTL$^\downarrow_1(\nextt, \release)$,
which is valid over data $\omega$-words iff
$\aformula$ is satisfiable over finite data words.
\end{proof}

\section{Concluding Remarks}

Satisfiability (over timed $\omega$-words)
for the safety fragment of metric temporal logic (MTL)
was shown decidable in \cite{Ouaknine&Worrell06b},
and nonelementary in \cite{Bouyeretal08} by reducing from
termination of channel machines with emptiness testing and insertion errors.
It would be interesting to investigate whether
ideas in the proof of Theorem~\ref{th:IPCANT} above can be combined with
those in the proof of primitive recursiveness of
termination of channel machines with occurrence testing and insertion errors
\cite{Bouyeretal08} to obtain that
satisfiability for safety MTL is primitive recursive.

Another open question is whether
nonemptiness of safety forward alternating tree automata with $1$ register
\cite{Jurdzinski&Lazic07} is primitive recursive.

\begin{acks}
I am grateful to St\'ephane Demri and James Worrell for helpful discussions.
\end{acks}

\bibliographystyle{acmtrans}
\bibliography{freeze_j}

\begin{received}
Received February 2008;
revised March 2009;
accepted April 2010
\end{received}

\end{document}